\documentclass{amsart}
\usepackage{amssymb, latexsym }

\newcommand{\beqa}{\begin{eqnarray*}}
\newcommand{\eeqa}{\end{eqnarray*}}
\newcommand{\beqn}{\begin{eqnarray}}
\newcommand{\eeqn}{\end{eqnarray}}

\newcommand{\R}{\mathbb R}

\newcommand{\G}{\Gamma}

\newcommand{\De}{\Delta}
\newcommand{\la}{\lambda}

\newcommand{\bs}{\boldsymbol}
\newcounter{cnt1}
\newcounter{cnt2}
\newcounter{cnt3}
\newcommand{\blr}{\begin{list}{$($\roman{cnt1}$)$}
 {\usecounter{cnt1} \setlength{\topsep}{0pt}
 \setlength{\itemsep}{0pt}}}
\newcommand{\bla}{\begin{list}{$($\alph{cnt2}$)$}
 {\usecounter{cnt2} \setlength{\topsep}{0pt}
 \setlength{\itemsep}{0pt}}}
\newcommand{\bln}{\begin{list}{$($\arabic{cnt3}$)$}
 {\usecounter{cnt3} \setlength{\topsep}{0pt}
 \setlength{\itemsep}{0pt}}}
\newcommand{\el}{\end{list}}
\newtheorem{thm}{Theorem}

\newtheorem{Def}[thm]{Definition}

\newtheorem{rem}[thm]{Remark}
\newcommand{\Rem}{\begin{rem} \rm}
\newcommand{\bdfn}{\begin{Def} \rm}
\newcommand{\edfn}{\end{Def}}

\newcommand{\ba}{\begin{array}}
\newcommand{\ea}{\end{array}}

\begin{document}
\title{\bf Two Mathematically Equivalent Versions of Maxwell's Equations}
\author[Gill]{Tepper L. Gill}
\address[Tepper L. Gill]{ Department of Electrical Engineering Howard University\\
Washington DC 20059 \\ USA, {\it E-mail~:} {\tt tgill@howard.edu}}
\author[Zachary]{Woodford W. Zachary}
\address[Woodford W. Zachary]{ Department of Electrical Engineering \\ Howard
University\\ Washington DC 20059 \\ USA, {\it E-mail~:} {\tt
wwzachary@earthlink.net}}
\begin{abstract}
This paper is a review of the canonical proper-time approach to relativistic mechanics and classical electrodynamics.  The purpose is to provide a physically complete classical background for a new approach to relativistic quantum theory. Here, we first show that there are two versions of Maxwell's equations.  The new version fixes the clock of the field source for all inertial observers.  However now, the (natural definition of the effective) speed of light is no longer an invariant for all observers, but depends on the motion of the source. This approach allows us to account for radiation reaction without the Lorentz-Dirac equation, self-energy (divergence), advanced potentials or any assumptions about the structure of the source.  The theory provides a new invariance group which, in general, is a nonlinear and nonlocal representation of the Lorentz group.  This approach also provides a natural (and unique) definition of simultaneity for all observers.

The corresponding particle theory is independent of particle number, noninvariant under time reversal (arrow of time), compatible with quantum mechanics and has a corresponding positive definite canonical Hamiltonian associated with the clock of the source. 

We also provide a brief review of our work on the foundational aspects of the corresponding relativistic quantum theory.  Here, we show that the standard square-root and the Dirac equations are actually two distinct  spin-$\tfrac{1}{2}$ particle equations.
\end{abstract}
\maketitle
\section*{\bf{Introduction}}
Einstein begins his 1905 [1], paper with the statement: ``It is known that Maxwell's electrodynamics as usually understood at the present time - - when applied to moving bodies, leads to asymmetries which do not appear to be inherent in the phenomena."  After quoting a few examples and the unsuccessful attempts of experimenters to discover the light medium (ether), he concludes that mechanics as well as electrodynamics possess no properties corresponding to absolute rest.  Thus, the laws of electrodynamics and optics will be valid for all frames in which the equations of mechanics holds.  He then suggests that we raise his conjecture to the status of a postulate called the ``principle of relativity''.  He then adds one other postulate to provide what is now known as the basic postulates of the special theory of relativity.

As noted by Bridgman [2], the special theory allows us to by-pass but not answer the fundamental question of ``the nature of the physical mechanism by which objects are lighted''.  From an operational point of view,  we must ask is it physically possible to consider light as a ``thing'' that travels?  Bridgman [2] observed that:
\begin{quote}
we can give no operational meaning to the idea that light exists at each point between source and sink.   The idea of light as a thing traveling is pure invention based on sense perceptions and the mechanical world view.
\end{quote}
\begin{rem}
Today, the generally accepted interpretation is that electromagnetic waves are formed when an electric field, $\bf E$, couples with a magnetic field, $\bf B$. (The electric and magnetic fields of an electromagnetic wave are perpendicular to each other and to the direction of the wave.)  The basic assumption is that these two fields move in the vacuum  and no ether is required (see Purcell [3]).
\end{rem}
However, we must be careful that we don't use the description of what these waves are composed of (i.e., solutions of the wave equation) as an interpretation of how they travel.
    
As noted in Miller [4], Einstein chose to consider light as a thing traveling for convenience.  This allowed him to use the standard notion of velocity for measurement purposes.  However, in the special theory, light is not a material particle nor is it a wave since, if it's a particle, its velocity cannot be independent of the source motion and, if it's a wave, it must travel in a medium (the ether), which is known to not have any effect on light! (Since the advent of quantum theory, many have assumed that the question is resolved by wave-particle duality, but this is far from true.)

It should not go unnoticed that, in a paper published almost at the same time, Einstein [5] used the concept of light as a ``localized energy packet'' to explain the photoelectric effect.  In fact, Planck [6] wrote: ``According to the latest statements by Einstein, it would be necessary to assume that free radiation in vacuum, and hence light waves themselves, has an atomistic constitution, and thus to abandon Maxwell's equations."  We should not be amazed at Planck's statement since, at the time,  the question of the need for Maxwell's equations at all was still an open subject.

In 1867, Ludvig Lorenz [7] introduced retarded vector and scalar potentials.   It was shown that these led to the same results obtained by Maxwell via the introduction of the displacement current into Amp\`{e}re's law.   Indeed, it has been known since then that all the results of the Maxwell theory can be obtained directly from the potentials, without ever introducing fields (i.e., action-at-a-distance).  It has recently been shown by Hamdan, Hariri and L\'{o}pez-Bonilla [8] that one can derive Maxwell's equations directly from the Lorentz force.  

There were many who took  L. Lorenz's position, but the major protagonist in this debate was Walther Ritz [9].  Ritz, like Einstein, accepted H. A. Lorentz's theory of the electron but rejected the ether.  He further noted that, from a strictly logical point of view, Maxwell's electric and magnetic fields, which appear to play such an important role, can be  entirely eliminated from the theory.  He argued that, in reality, Maxwell's theory deals only with certain relations between space and time.  In his view, we could simply return to the elementary actions (retarded potentials).  He further pointed out that the field equations have an infinite number of solutions that are incompatible with experiment and, in order to eliminate these extraneous solutions, it is necessary to adopt the retarded potentials anyway.  This introduces an additional assumption which is not needed if we start with the retarded potentials in the first place.

Einstein did not completely accept, but was swayed by Ritz's position.  Indeed, in his 1909 paper [10], Einstein stated:
\begin{quote}
According to the usual theory, an oscillating ion generates a divergent spherical wave.  The reverse process does not exist as a elementary process.  The convergent wave is indeed mathematically possible; but for its approximate realisation an enormous number of elementary emitting systems would be required. Hence, the elementary process of light-emission has not as such the character of reversibility.  Herein, I believe, our wave theory is incorrect.  It seems in relation to this point Newton's emission theory contains more truth than the wave theory, for the energy communicated to a light-particle in emission is not spread over infinite space but remains available for an elementary process of absorption.
\end{quote}
Here, Einstein is agreeing with Ritz's position that retarded potentials express the elementary process of emission, whereas Maxwell's equations do not.  We get a further clue to Einstein's thinking on this subject from his {\it Autobiographical notes} of 1949 [11] (see Brown [12]).  
\begin{quote}
Reflections of this type made it clear to me as long ago as 1900, i.e., shortly after Planck's trailblazing work, that neither mechanics nor electrodynamics could (except in limiting cases) claim exact validity.
\end{quote}
Brown points out that, because he was not sure that Maxwell's theory would survive the existence of photons, Einstein had the foresight to derive the Lorentz transformations from kinematical arguments, as opposed to the symmetry properties of Maxwell's equations.  (He believed that the Lorentz transformations were fundamental and would survive any failures in the Maxwell theory.) 
 
There always was a certain tension between field theory and action-at-a-distance.  The most famous recent work on the subject is the Wheeler-Feynman formulation of classical electrodynamics [13], in which they eliminate the field completely in favor of action-at-a-distance.  This allowed them to solve the self-energy divergence problem associated with the then-accepted Dirac theory [14] (for details, see [18]).  However, among other things, the need for both advanced and retarded interactions, the inability to quantize and the intrinsic usefulness of the self-energy divergence for the success of quantum electrodynamics became important reasons for its lack of favor as a replacement for the Dirac approach.
\subsection*{\bf Purpose:\;}  
In this paper, we introduce a canonical proper-time formulation of classical electrodynamics, where the local clock of the moving system replaces the clock of the observer. This approach is mathematically equivalent, but is not physically equivalent to the conventional approach.  Physically, this change is equivalent to a new definition of velocity for relativistic systems and by-passes all the well-known problems with the conventional theory.  We also develop the corresponding particle theory, which produces a positive definite canonical Hamiltonian; the local clock is non-invariant under time-reversal (time-arrow).  As an application, we briefly look at a few areas where additional assumptions are required in order to restore the conventional theory and explain new findings.  The main purpose is to prepare the way, on the classical level, for the corresponding (relativistic) quantum theory.  However, we briefly discuss the quantum foundations and our current progress in this direction.  
\section{\bf Maxwell Theory}
\subsection{{Foundations}}
The basis for all of natural science is the belief that there is a real, external world, which exists independent of our perceptions of it. The first postulate is an outgrowth of the fundamental question of what can we know about this external world and how do we know we can actually establish that we know it?  The answer to this question for physics is to explicitly state the conditions under which we are willing to call a particular mathematical model a physical law, namely, if it satisfies the principle of relativity.  

A major objective of theoretical physics is to provide a mathematical model (using a minimal number of variables) to describe the cause and effect relationships observed in the experiment.  If this model is physically reliable, mathematically consistent and leads to insights and/or predictions that were heretofore unsuspected, it may, over time, acquire the status of a physical law.  The first postulate of the special theory of relativity imposes a natural constraint on the extent that we may believe in the law and the results of related experiments; namely that, any other observer (not necessarily human), in any other galaxy using similar equipment, in any other inertial frame of reference, must be able to obtain similar results and model, with differences accounted for by a Lorentz transformation.  Thus, a true physical law must be independent of the particular class of sentient beings discovering it.  

\subsection{Local-time Maxwell Equations}
In his formulation, it was natural for Einstein to use the clock of the observer to measure time. In the following section, we show that an equally valid clock to use is the local clock of the observed system, which is generally known as the proper-time.  (In this terminology, the conventional clock used is the proper-time of the observer.) 

In order to formulate the local-time version of Maxwell's equations, it is convenient to start with the standard definition of proper-time:
\[
d\tau ^2  = dt^2  - \frac{1}{{c^2 }}d{\mathbf{x}}^2  = dt^2 \left[ {1 - \frac{{{\mathbf{w}}^2 }}{{c^2 }}} \right],\quad {\mathbf{w}} = \frac{{d{\mathbf{x}}}}{{dt}}.
\]
Motivated by geometry and the mathematical philosophy of the era, Minkowski introduced the concept of proper time (first discovered by Poincar\'{e}).  Recently, it has been suggested by Damour [15]  that Minkowski was not aware that $d{\tau}$ is not an exact one-form and hence cannot be used for a metric. It is clear that Minkowski became aware of this fact (see Sommerfeld's notes in  [16] after the translation of Minkowski's paper (pg. 94)).   (For very interesting additional discussion on this and other related points, see Walters [17] and included references.)
   
Nevertheless, some have dismissed this (physical) fact by attaching a ``co-moving observer" on the tangent curve (bundle) of the moving particle in order to induce an instantaneous exact one-form for the four-geometry at each time slice. This is mathematically correct but physically invalid, since for example, it is impossible for such an observer to sit on the tangent curve of a high-energy muon entering our atmosphere, or an antiproton in the accelerator at Fermi Laboratory.  In fact, the use of time dilation in order to explain results that have no explanation using the observer's measuring rods and clock, must be considered incorrect at a basic methodological level.
 
However, there is an important physical reason why $d{\tau}$ is not an exact  one-form. Physically, a particle can traverse many different paths (in space) during any given $\tau$ interval.  This reflects the fact that the distance traveled in a given time interval depends on the forces acting on the particle.  This suggests that the actual clock of the source carries additional physical information about the acting forces.    In order to see that this is the case, rewrite the above equation as:
\[
dt^2  = d\tau ^2  + \frac{1}{{c^2 }}d{\mathbf{x}}^2  = d\tau ^2 \left[ {1 + \frac{{{\mathbf{u}}^2 }}{{c^2 }}} \right],\quad {\mathbf{u}} = \frac{{d{\mathbf{x}}}}
{{d\tau }}.
\]
For any other observer, we have:
\[
dt'^2  = d\tau ^2  + \frac{1}{{c^2 }}d{\mathbf{x}}'^2  = d\tau ^2 \left[ {1 + \frac{{{\mathbf{u}}'^2 }}{{c^2 }}} \right],\quad {\mathbf{u}}' = \frac{{d{\mathbf{x}}'}}
{{d\tau }}.
\]
Thus, all observers can use one unique clock to discuss all events associated with the source (simultaneity). It should also be noted that this also corresponds to a change in our (independent) configuration space variables from $( {{\bf x}(t)}, {{\bf w}(t)}) \rightarrow ({{\bf x}(\tau)}, {{\bf u}(\tau)})$. (However, because the corresponding momentum ${\mathbf{p}} = m{\mathbf{w}} = m_0 {\mathbf{u}}$ where $m = \gamma m_0 $, the phase space variables do not change.)  Returning to the first equation, we see that the new metric defined by $dt$ is clearly exact, while the representation space is now Euclidean.  However, the natural definition of velocity is no longer ${\bf{w}}=d{\bf{x}}/dt$ but ${\bf{u}}=d{\bf{x}}/d{\tau}$. This fact suggests that there may be a certain duality in the relationship between $t,\;\tau$ and ${\bf{w}},\;{\bf{u}}$.  To see that this is indeed the case, recall that ${\mathbf{u}} = {{\mathbf{w}} \mathord{\left/
 {\vphantom {{\mathbf{w}} {\sqrt {1 - \left( {{{{\mathbf{w}}^2 } \mathord{\left/
 {\vphantom {{{\mathbf{w}}^2 } {c^2 }}} \right. \kern-\nulldelimiterspace} {c^2 }}} \right)} }}} \right. \kern-\nulldelimiterspace} {\sqrt {1 - \left( {{{{\mathbf{w}}^2 } \mathord{\left/
 {\vphantom {{{\mathbf{w}}^2 } {c^2 }}} \right. \kern-\nulldelimiterspace} {c^2 }}} \right)} }}$.  Solving for ${\bf{w}}$, we get that ${\mathbf{w}} = {{\mathbf{u}} \mathord{\left/ {\vphantom {{\mathbf{u}} {\sqrt {1 + \left( {{{{\mathbf{u}}^2 } \mathord{\left/ {\vphantom {{{\mathbf{u}}^2 } {c^2 }}} \right. \kern-\nulldelimiterspace} {c^2 }}} \right)} }}} \right. \kern-\nulldelimiterspace} {\sqrt {1 + \left( {{{{\mathbf{u}}^2 } \mathord{\left/ {\vphantom {{{\mathbf{u}}^2 } {c^2 }}} \right. \kern-\nulldelimiterspace} {c^2 }}} \right)} }}$.    If we set $b = \sqrt {c^2  + {\mathbf{u}}^2 }$, this relationship can be written as  
\beqn
\frac{{\mathbf{w}}}{c} = \frac{{\mathbf{u}}}{b}.
\eeqn
For reasons to be clear momentarily, we call $b$ the collaborative speed of light.  Indeed, we see that 
\beqn
\frac{1}{c}\frac{\partial }{{\partial t}} = \frac{1}{c}\frac{{\partial \tau }}
{{\partial t}}\frac{\partial }{{\partial \tau }} = \frac{1}{c}\frac{1}
{{\sqrt {1 + \left( {{{{\mathbf{u}}^2 } \mathord{\left/{\vphantom {{{\mathbf{u}}^2 } {c^2 }}} \right.\kern-\nulldelimiterspace} {c^2 }}} \right)} }}\frac{\partial }{{\partial \tau }} = \frac{1}{b}\frac{\partial }{{\partial \tau }}.
\eeqn
For any other observer, it is easy to see that the corresponding result will be:
\beqa
\frac{{{\bf{w}}'}}
{c} = \frac{{{\bf{u}}'}}
{{b'}},\quad \quad \frac{1}
{c}\frac{\partial }
{{\partial t'}} = \frac{1}
{{b'}}\frac{\partial }
{{\partial \tau }}.
\eeqa
We see from the above two equations that the non-invariance of $t$ and the invariance of $c$ on the left is replaced by the non-invariance of $b$ and the invariance of $\tau$ on the right.  These equations clearly represent mathematically equivalent relations (i.e., they are identities). Thus, wherever they are used consistently as replacements for each other, they can't change the mathematical relationships. In order to see their impact on Maxwell's equations, write them (in c.g.s. units)
\beqn
\begin{gathered}
\nabla  \cdot {\mathbf{B}} = 0,\quad \quad \quad \nabla  \cdot {\mathbf{E}} = 4\pi \rho , \hfill \\
\nabla  \times {\mathbf{E}} =  - \frac{1}{c}\frac{{\partial {\mathbf{B}}}}{{\partial t}},\quad \nabla  \times {\mathbf{B}} = \frac{1}{c}\left[ {\frac{{\partial {\mathbf{E}}}}
{{\partial t}} + 4\pi \rho {\mathbf{w}}} \right]. \hfill \\ 
\end{gathered} 
\eeqn
Using equations (1) and (2) in (3), we have ({{\it{the identical mathematical representation for Maxwell's equations}}):
\beqa
\begin{gathered}
\nabla  \cdot {\mathbf{B}} = 0,\quad \quad \quad \nabla  \cdot {\mathbf{E}} = 4\pi \rho , \hfill \\
\nabla  \times {\mathbf{E}} =  - \frac{1}{b}\frac{{\partial {\mathbf{B}}}}{{\partial \tau }},\quad \nabla  \times {\mathbf{B}} = \frac{1}{b}\left[ {\frac{{\partial {\mathbf{E}}}}
{{\partial \tau }} + 4\pi \rho {\mathbf{u}}} \right]. \hfill \\ 
\end{gathered} 
\eeqa
Thus, we see that Maxwell's equations are equally valid when the local time of the particle is used to describe the fields.  This leads to the following conclusions:
\begin{enumerate}
\item There are two distinct clocks to use in the representation of Maxwell's equations.  Thus, the choice of clocks is a convention in the true sense of Poincar\'{e} (see conclusion). 
\item Since the two representations are mathematically equivalent, we conclude that mathematical equivalence is not always physical equivalence.   (This will be clear after we derive the corresponding wave equation below.)  
\item When the local clock of the system is used, the constant speed of light $c$ is replaced by the effective speed of light $b$, which depends on the motion of the system (i.e., $b = \sqrt {c^2  + {\mathbf{u}}^2 }$).
\item  There is another group (closely related to the Lorentz group) which fixes the local-time of the particle for all observers. 
\end{enumerate}
Before constructing the proper-time group, we derive the corresponding wave equations in the local-time variable.  Taking the curl of the last two equations (above), and using standard vector identities, we get: 
\beqn
\begin{gathered}
 \frac{1}{{b^2 }}\frac{{\partial^2 {\mathbf{B}}}}{{\partial \tau ^2 }} - \frac{{{\mathbf{u}} \cdot {\mathbf{a}}}}{{b^4 }}\left[ {\frac{{\partial {\mathbf{B}}}}
{{\partial \tau }}} \right] - \nabla ^2 \cdot {\mathbf{B}} = \frac{1}{b}\left[ 4\pi \nabla  \times (\rho {\mathbf{u}}) \right], \hfill \\
 \frac{1}{{b^2 }}\frac{{\partial ^2 {\mathbf{E}}}}{{\partial \tau ^2 }} - \frac{{{\mathbf{u}} \cdot {\mathbf{a}}}}{{b^4 }}\left[ {\frac{{\partial {\mathbf{E}}}}
{{\partial \tau }}} \right] - \nabla ^2  \cdot {\mathbf{E}} =  - \nabla (4\pi \rho ) - \frac{1}{b}\frac{\partial }{{\partial \tau }}\left[ {\frac{{4\pi (\rho {\mathbf{u}})}}{b}} \right], \hfill \\ 
\end{gathered} 
\eeqn				
where ${\bf{a}} = d{\bf{u}}/d\tau$ is the effective acceleration caused by external forces.  Thus, we see that a new term arises when the proper-time of the source is used to describe the fields.   This makes it clear that the local clock encodes information about the particle's interaction that is unavailable when the clock of the observer is used to describe the fields, and shows clearly that physical equivalence is not the same as mathematical equivalence.   The new term in equation (4) is dissipative, acts to oppose the acceleration, is zero when ${\bf{a}} =0$  and, arises instantaneously with the action of forces on the particle. {\it Furthermore, as expected, this term does not depend on the nature of the force causing the acceleration of the charged particle.} This is exactly what one expects of the back reaction caused by the inertial resistance of the particle to accelerated motion and, according to Wheeler and Feynman [13], is precisely what is meant by radiation reaction.  Thus, the collaborative use of the observer's coordinate system and the local clock of the observed system provides intrinsic information about the field dynamics not available in the conventional formulation of Maxwell's theory.  If we make a scale transformation (at fixed position) with ${\bf{E}} \to (b/c)^{1/2}{\bf{E}}$    and ${\bf{B}} \to (b/c)^{1/2}{\bf{B}}$, the equations in (4) transform to 
\beqn
\begin{gathered}
 \frac{1}{{b^2 }}\frac{{\partial ^2 {\mathbf{B}}}}{{\partial \tau ^2 }} - {\text{ }}\nabla ^2 {\kern 1pt}  \cdot {\mathbf{B}} + \left[ \frac{{\ddot b}}
{{2b^3 }}-{\frac{{3\dot b^2 }}{{4b^4 }}  } \right]{\mathbf{B}} = \frac{{c^{1/2} }}{{b^{3/2} }}\left[ {4\pi \nabla  \times (\rho {\mathbf{u}})} \right], \hfill \\
 \frac{1}{{b^2 }}\frac{{\partial ^2 {\mathbf{E}}}}{{\partial \tau ^2 }} - {\text{ }}\nabla ^2 {\kern 1pt}  \cdot {\mathbf{E}} +  \left[ \frac{{\ddot b}}
{{2b^3 }}-{\frac{{3\dot b^2 }}{{4b^4 }}  } \right]{\mathbf{E}} =  - \frac{{c^{1/2} }}{{b^{1/2} }}\nabla (4\pi \rho ) - \frac{{c^{1/2} }}{{b^{3/2}}}\frac{\partial}{{\partial \tau }}\left[ {\frac{{4\pi (\rho {\mathbf{u}})}}{b}} \right]. \hfill \\ 
\end{gathered} 
\eeqn
This is the Klein-Gordon equation with an effective mass $\mu$ given by
\beqn
\mu  = \left\{ {\frac{{\hbar ^2 }}{{c^2 }}\left[ {\frac{{\ddot b}}
{{2b^3 }} - \frac{{3\dot b^2 }}{{4b^4 }}} \right]} \right\}^{1/2}  = \left\{ {\frac{{\hbar ^2 }}{{c^2 }}\left[ {\frac{{{\mathbf{u}} \cdot {\mathbf{\ddot u}} + {\mathbf{\dot u}}^2 }}  {{2b^4 }} - \frac{{5\left( {{\mathbf{u}} \cdot {\mathbf{\dot u}}} \right)^2 }}
{{4b^6 }}} \right]} \right\}^{1/2}. 
\eeqn
Let $({\bf{x}}(\tau),\tau)$ represent the field position and $(\bar{\bf{x}}(\tau'),\tau')$ the retarded position of a source charge $e$, with ${\bf r}={ \bf{x}}-\bar{\bf{x}}$.  If we set $r=\left|{ \bf{x}}-\bar{\bf{x}}\right|$, $s=r-(\tfrac{({\bf r}\cdot{\bf u})}{b})$, and ${\bf r}_{\bf u}={\bf r}-{\tfrac{r}{b}}{\bf u}$, then we were able to compute the $\bf E$ and $\bf B$ fields directly in [18]  to obtain:
\beqn
{\mathbf{E}}({\mathbf{x}},\tau ) = \frac{{e\left[ {{\mathbf{r}}_{\mathbf{u}} (1 - {{{\mathbf{u}}^2 } \mathord{\left/
 {\vphantom {{{\mathbf{u}}^2 } {b^2 }}} \right.
 \kern-\nulldelimiterspace} {b^2 }})} \right]}}
{{s^3 }} + \frac{{e\left[ {{\mathbf{r}} \times ({\mathbf{r}}_{\mathbf{u}}  \times {\mathbf{a}})} \right]}}
{{b^2 s^3 }} + \frac{{e({\mathbf{u}} \cdot {\mathbf{a}})\left[ {{\mathbf{r}} \times ({\mathbf{u}} \times {\mathbf{r}})} \right]}}
{{b^4 s^3 }}
\eeqn
and
\[
{\mathbf{B}}({\mathbf{x}},\tau ) = \frac{{e\left[ {({\mathbf{r}} \times {\mathbf{r}}_{\mathbf{u}} )(1 - {{{\mathbf{u}}^2 } \mathord{\left/
 {\vphantom {{{\mathbf{u}}^2 } {b^2 }}} \right.
 \kern-\nulldelimiterspace} {b^2 }})} \right]}}
{{rs^3 }} + \frac{{e{\mathbf{r}} \times \left[ {{\mathbf{r}} \times ({\mathbf{r}}_{\mathbf{u}}  \times {\mathbf{a}})} \right]}}
{{rb^2 s^3 }} + \frac{{er({\mathbf{u}} \cdot {\mathbf{a}})({\mathbf{r}} \times {\mathbf{u}})}}
{{b^4 s^3 }}.
\]
(It is easy to see that $\bf B$ is orthogonal to $\bf E$.)  The first two terms in the above equations are standard, in the $({\bf{x}}(t),{\bf{w}}(t))$ variables.  The third part of both equations is new and arises because of the dissipative term in our wave equation.   It is easy to see that ${\bf r}\times({\bf u}\times{\bf r})=r^2{\bf u}-({\bf u}\cdot{\bf r}){\bf r}$, so we get a component along the direction of motion. (Thus,  the $\bf E$ field has a  longitudinal part.) This confirms our claim that the new dissipative term is equivalent to an effective mass that arises due to the collaborative acceleration of the particle.   This means that the cause for radiation reaction comes directly from the use of the local clock to formulate Maxwell's equations.  Thus, in this approach, there is no need to assume advanced potentials, self-interaction, mass renormalization along with the Lorentz-Dirac equation in order to account for it (radiation reaction), as has been required when the observer clock is used. Furthermore, no assumptions about the structure of the source are required.  
\subsection{\bf Proper-time Group} 
We now identify the new  (proper-time) transformation group that preserves the first postulate of the special theory. The standard (Lorentz) time transformations between two inertial observers can be written as
\beqn
t' = \gamma ({\mathbf{v}})\left[ {t - {{{\mathbf{x}} \cdot {\mathbf{v}}} \mathord{\left/ {\vphantom {{{\mathbf{x}} \cdot {\mathbf{v}}} {c^2 }}} \right.
 \kern-\nulldelimiterspace} {c^2 }}} \right], \quad \quad \quad \quad {\text{      }}t = \gamma ({\mathbf{v}})\left[ {t' + {{{\mathbf{x'}} \cdot {\mathbf{v}}} \mathord{\left/
 {\vphantom {{{\mathbf{x'}} \cdot {\mathbf{v}}} {c^2 }}} \right. \kern-\nulldelimiterspace} {c^2 }}} \right].
\eeqn
We want to replace $t,\;t'$ by $\tau$.  To do this, use the relationship between $dt$ and $d\tau$ to get:
\beqn
t = \tfrac{1}{c}\int_0^\tau  {b(s)} ds = \tfrac{1}{c}\bar b\tau ,\quad  t' = \tfrac{1}{c}\int_0^\tau  {b'(s)} ds = \tfrac{1}{c}\bar b'\tau, 
\eeqn
where we have used the mean value theorem of calculus to obtain the end result, so that both $\bar b$ and $\bar b'$ represent an earlier $\tau$-value of $b$ and $b'$ respectively.  Note that, as $b$ and $b'$ depend on $\tau$, the transformations (9) represent explicit nonlinear relationships between $t,\;t'$ and $\tau$ (during interaction).  (This is to be expected in the general case when the system is acted on by external forces.)   If we set
\[
{\mathbf{d}}^ *   = {{\mathbf{d}} \mathord{\left/ {\vphantom {{\mathbf{d}} {\gamma ({\mathbf{v}})}}} \right.\kern-\nulldelimiterspace} {\gamma ({\mathbf{v}})}} - (1 - \gamma ({\mathbf{v}}))\left[ {({{{\mathbf{v}} \cdot {\mathbf{d}})} \mathord{\left/
{\vphantom {{{\mathbf{v}} \cdot {\mathbf{d}})} {(\gamma ({\mathbf{v}}){\mathbf{v}}^2 }}} \right. \kern-\nulldelimiterspace} {(\gamma ({\mathbf{v}}){\mathbf{v}}^2 }})} \right]{\mathbf{v}},
\] 
we can write the transformations that fix $\tau$ as:
\beqn
\begin{gathered}
\quad {\mathbf{x'}} = \gamma ({\mathbf{v}})\left[ {{\mathbf{x}}^ *   - ({{\mathbf{v}} \mathord{\left/{\vphantom {{\mathbf{v}} c}} \right.\kern-\nulldelimiterspace} c})\bar b\tau } \right],\quad \quad \quad \quad \,{\mathbf{x}} = \gamma ({\mathbf{v}})\left[ {{\mathbf{x'}}^*   + ({{\mathbf{v}} \mathord{\left/{\vphantom {{\mathbf{v}} c}} \right.\kern-\nulldelimiterspace} c})\bar b'\tau } \right], \hfill \\
\quad {\mathbf{u'}} = \gamma ({\mathbf{v}})\left[ {{\mathbf{u}}^ *   - ({{\mathbf{v}} \mathord{\left/{\vphantom {{\mathbf{v}} c}} \right.\kern-\nulldelimiterspace} c})b} \right],\quad \quad \quad \quad \; \,{\text{   }}{\mathbf{u}} = \gamma ({\mathbf{v}})\left[ {{\mathbf{u'}}^*   + ({{\mathbf{v}} \mathord{\left/ {\vphantom {{\mathbf{v}} c}} \right. \kern-\nulldelimiterspace} c})b'} \right], \hfill \\
\quad {\mathbf{a'}} = \gamma ({\mathbf{v}})\left\{ {{\mathbf{a}}^*  - {\mathbf{v}}\left[ {({{{\mathbf{u}} \cdot {\mathbf{a}})} \mathord{\left/{\vphantom {{{\mathbf{u}} \cdot {\mathbf{a}})} {(bc}}} \right.\kern-\nulldelimiterspace} {(bc}})} \right]} \right\},\quad {\text{   }}{\mathbf{a}} = \gamma ({\mathbf{v}})\left\{ {{\mathbf{a'}}^*   + {\mathbf{v}}\left[ {({{{\mathbf{u'}} \cdot {\mathbf{a'}})} \mathord{\left/{\vphantom {{{\mathbf{u'}} \cdot {\mathbf{a'}})} {(b'c}}} \right. \kern-\nulldelimiterspace} {(b'c}})} \right]} \right\}. \hfill \\ 
\end{gathered} 
\eeqn
If we put equation (9) in (8), differentiate with respect to $\tau$  and cancel the extra factor of $c$, we get the transformations between $b$ and $b'$:
\beqn
b'(\tau ) = \gamma ({\mathbf{v}})\left[ {b(\tau ) - {{{\mathbf{u}} \cdot {\mathbf{v}}} \mathord{\left/{\vphantom {{{\mathbf{u}} \cdot {\mathbf{v}}} c}} \right. \kern-\nulldelimiterspace} c}} \right],\quad \quad \quad \quad {\text{ }}b(\tau ) = \gamma ({\mathbf{v}})\left[ {b'(\tau ) + {{{\mathbf{u'}} \cdot {\mathbf{v}}} \mathord{\left/
 {\vphantom {{{\mathbf{u'}} \cdot {\mathbf{v}}} c}} \right. \kern-\nulldelimiterspace} c}} \right].
\eeqn
Equations (10) in (11) provide an explicit nonlinear representation of the Lorentz group, which uses the local clock to describe the dynamics of the system and preserves the first postulate of the special theory (the only one that really matters). We call this group the proper-time group.  

It was shown in [18] that Maxwell's equations transform in the same way as in the conventional theory.  However, the charge and current density have the following transformations:  
\beqn
 {\bf J}'={\bf J}+(\gamma -1){{({\bf J}\cdot {\bf v})} \over {{\bf
v}^2}}v-\gamma {b \over c}\rho {\bf v},  
\eeqn   
\beqn
b'\rho' =\gamma ({\bf
v})\left[ {b\rho -({{{\bf J}\cdot {\bf v}} \mathord{\left/ {\vphantom {{{\bf J}\cdot
{\bf v}} c}} \right.\kern-\nulldelimiterspace} c})} \right]. 
\eeqn  
Using the first equation of (11) in (13), we get:  
\beqn
 \rho' ={{\rho -({{{\bf
J}\cdot {\bf v}}\mathord{\left/ {\vphantom {{{\bf J}\cdot {\bf v}} {bc}}}\right.
\kern-\nulldelimiterspace} {bc}})}\over {1-({{{\bf u}\cdot {\bf v}} \mathord{\left/
{\vphantom {{{\bf u}\cdot {\bf v}} {bc}}} \right. \kern- \nulldelimiterspace}
{bc}})}}.  
\eeqn   
This result is different from the standard one, (which we obtain if we set $b'=b=c$ in (13)),  
$$ 
\rho' =\gamma ({\bf v})\left[ {\rho
-({{{\bf J}\cdot {\bf v}} \mathord{\left/ {\vphantom {{{\bf J}\cdot {\bf v}} {c^2}}}
\right.\kern-\nulldelimiterspace} {c^2}})} \right]. 
$$    
Furthermore, if we insert the expression ${{\bf J}/c}={\bf \rho} ({{\bf
u}/b})$ in (14); 
we obtain  
\beqn
\rho' =\rho {{1-({{{\bf u}\cdot {\bf v}} \mathord{\left/ {\vphantom {{{\bf u}\cdot
{\bf v}} {b^2}}} \right. \kern-\nulldelimiterspace} {b^2}})} \over {1-({{{\bf u}\cdot
{\bf v}} \mathord{\left/ {\vphantom {{{\bf u}\cdot {\bf v}} {bc}}} \right. \kern-
\nulldelimiterspace} {bc}})}}. 
\eeqn 
In order to see the impact of equation (15), suppose that a spherical charge distribution is at rest in the unprimed frame. From (15), we see that ${\bf u}=0$, so that $\rho'=\rho$.  Since the primed frame is arbitrary, the charge distribution will appear spherical to all observers.  This is what we would expect on physical grounds, so that relatively moving frames should not change the symmetry properties of charged objects.  In particular, a spherical charge distribution in one frame should not be distorted in another moving frame (i.e., display physical effects due to another observer's relative motion).   
\subsection{\bf Canonical Proper-Time Particle Theory}
We now investigate the corresponding particle theory.  Since we desire complete compatibility with quantum theory, it is natural to require that  any change from the observer clock to the local clock of the observed system be a canonical change of variables. The key concept to our approach may be seen by examining the time evolution of a dynamical parameter $W({\bf{x}},{\bf{p}})$, via the standard formulation of classical mechanics, described in terms of the Poisson brackets:
\beqn
\frac{{dW}}{{dt}} = \left\{ {H,W} \right\}.
\eeqn
We can also represent the dynamics using the proper (or local) time of the system by using the representation $d\tau  = ({1 \mathord{\left/{\vphantom {1 \gamma }} \right. \kern-\nulldelimiterspace} \gamma })dt = ({{mc^2 } \mathord{\left/ {\vphantom {{mc^2 } H}} \right. \kern-\nulldelimiterspace} H})dt$,
so that:
\[
\frac{{dW}}{{d\tau }} = \frac{{dt}}{{d\tau }}\frac{{dW}}{{dt}} = \frac{H}
{{mc^2 }}\left\{ {H,W} \right\}.
\]
Assuming a well-defined (invariant) rest energy ($mc^2$) for the system, we determine the canonical proper-time Hamiltonian $K$ such that:
\[
\left\{ {K,W} \right\} = \frac{H}{{mc^2 }}\left\{ {H,W} \right\},\quad \left. K \right|_{{\mathbf{p}} = 0}  = \left. H \right|_{{\mathbf{p}} = 0}  = mc^2. 
\]
Using 
\[
\begin{gathered}
  \left\{ {K,W} \right\} = \left[ {\frac{H}{{mc^2 }}\frac{{\partial H}}{{\partial {\mathbf{p}}}}} \right]\frac{{\partial W}}{{\partial {\mathbf{x}}}} - \left[ {\frac{H}
{{mc^2 }}\frac{{\partial H}}{{\partial {\mathbf{x}}}}} \right]\frac{{\partial W}}
{{\partial {\mathbf{p}}}} \hfill \\
  {\text{          }} = \frac{\partial }{{\partial {\mathbf{p}}}}\left[ {\frac{{H^2 }}
{{2mc^2 }} + a} \right]\frac{{\partial W}}{{\partial {\mathbf{x}}}} - \frac{\partial }
{{\partial {\mathbf{x}}}}\left[ {\frac{{H^2 }}{{2mc^2 }} + a'} \right]\frac{{\partial W}}
{{\partial {\mathbf{p}}}}, \hfill \\ 
\end{gathered} 
\]
we get that $a = a' = \tfrac{1}{2}mc^2$, so that (assuming no explicit time dependence)
\[
K = \frac{{H^2 }}{{2mc^2 }} + \frac{{mc^2 }}{2},\quad {\text{and }}\quad \frac{{dW}}
{{d\tau }} = \left\{ {K,W} \right\}.
\] 
Since $\tau$ is invariant during interaction (minimal coupling), we make the natural assumption that (the form of) $K$ also remains invariant.  Thus, if $\sqrt {c^2 {\mathbf{p}}^2  + m^2 c^4 }  \to \sqrt {c^2 {\bs{\pi }}^2  + m^2 c^4 }  + V$, where  $\bf A$ a vector potential, $V$ is a potential energy term and  $\pi  = {\mathbf{p}} - \tfrac{e}
{c}{\mathbf{A}}$.  In this case, $K$ becomes:
\[
K=\frac{{ {\bs{\pi }}^2}}
{{2m}} + mc^2  + \frac{{V^2 }}
{{2mc^2 }} + \frac{{V\sqrt {c^2  {\bs{\pi }}^2+ m^2 c^4 } }}
{{mc^2 }}.
\]
If we set $H_0=\sqrt {c^2 {\bs{\pi }}^2  + m^2 c^4 }$, use standard vector identities with $H_0=mcb$, $\nabla \times \bs\pi=-\tfrac{e}{c}\bf{B}$, and compute Hamilton's equations, we get:
\beqn
\begin{gathered}
  {\bf{u}} = \frac{{d{\bf{x}}}}
{{d\tau }} = \left[ {1 + \frac{V}
{{H_0 }}} \right]\frac{\bs\pi }
{m} = \frac{\bs \pi }
{{\tilde m}}, \; {\tilde m}=\left[ {1 + \frac{V}
{{H_0 }}} \right]^{-1}m, \hfill \\
  \frac{{d{\bf{p}}}}
{{d\tau }} =  - \frac{{\left[ {\left( {\bs \pi  \cdot \nabla } \right) \bs \pi  + \tfrac{e}
{c}\bs \pi  \times {\bf{B}}} \right]}}
{m}\left[ {1 + \frac{V}
{{H_0 }}} \right] - \nabla V\frac{{H_0 }}
{{mc^2 }}\left[ {1 + \frac{V}
{{H_0}}} \right] \hfill \\
  \quad \quad \quad  = \tfrac{e}
{c}  \left( {{\bf{u}} \cdot \nabla } \right) {\bf{A}}  + \tfrac{e}
{c}{\bf{u}} \times {\bf{B}} - \nabla V\frac{b}
{c}\left[ {1 + \frac{V}
{{H_0 }}} \right]. \hfill \\ 
\end{gathered} 
\eeqn
We remark that one can view $\tilde{m}$ as a (finite) renormalization of $m$ which occurs the moment that the potential is turned on.  This may seem strange to one not familiar with the history of this subject.  An excellent discussion of renormalization from a historical perspective can be found in the article by Dresden [19].  
 
In order to see the impact of our condition that $K$ remains invariant during interaction in another way, compute the Lagrangian from $Ld\tau  = {\bf{p}} \cdot d{\bf{x}} - Kd\tau $, to get:
\[
L = \tilde m{\bf{u}}^2  - \frac{{\tilde m{\bf{u}}^2 }}
{2}\left( {\frac{{\tilde m}}
{m}} \right) - mc^2  - \frac{{V^2 }}
{{2mc^2 }} - V\left( {\frac{b}
{c}} \right) + \frac{e}
{c}{\bf{A}} \cdot {\bf{u}}.
\]
However, if we use the fact that ${\bs{\pi}}= \tilde{m}{\bf{u}}$ directly in $H_0$, we get the implicit relation $b=\sqrt{c^2+\frac{  {\tilde{m}}^2{\bf{u}}^2  }{m^2}}$.  If we use this in our equation for the metric, we get (in spherical coordinates):
\[
\begin{gathered}
  dt^2  = \left[ {1 + \frac{{{\bf{u}}^2 }}
{{c^2 \left[ {1 + \tfrac{V}
{{H_0}}} \right]^2 }}} \right]d\tau ^2  \Rightarrow  \hfill \\
  c^2 dt^2  = c^2 d\tau ^2  + \frac{{d{\bf{x}}^2 }}
{{\left[ {1 + \tfrac{V}
{{H_0 }}} \right]^2 }} = c^2 d\tau ^2  + \frac{{dr^2 }}
{{\left[ {1 + \tfrac{V}
{{H_0 }}} \right]^2 }} + \frac{{r^2 d\theta ^2 }}
{{\left[ {1 + \tfrac{V}
{{H_0}}} \right]^2 }} + \frac{{r^2 \sin ^2 \theta d\phi ^2 }}
{{\left[ {1 + \tfrac{V}
{{H_0 }}} \right]^2 }}. \hfill \\ 
\end{gathered} 
\]
Thus, we see that the metric becomes deformed in the presence of a potential (i.e., geometry is created by physics).

If we multiply equation (17) by $b/c$, compute the total derivative of $\bf{A}$ with respect to $\tau$ and use the definition of $\bf{E}$, with $V = e\Phi $, we have:
\beqn 
 \begin{gathered}
  \frac{c}
{b}\left[ {\frac{{d{\mathbf{p}}}}
{{d\tau }} - \frac{e}
{c}\frac{{d{\mathbf{A}}}}
{{d\tau }}} \right] =  - \frac{e}
{b}\frac{{\partial {\mathbf{A}}}}
{{\partial \tau }} + \tfrac{e}
{b}{\mathbf{u}} \times {\mathbf{B}} - e\nabla \Phi \left[ {1 + \frac{V}
{{mcb }}} \right] \hfill \\
  \quad \quad \quad \quad \quad \quad  = e{\mathbf{E}} + \tfrac{e}
{b}{\mathbf{u}} \times {\mathbf{B}} -e \nabla \Phi \frac{V}
{{mcb }}. \hfill \\ 
\end{gathered}
\eeqn
It has been observed by Feynman [20] that, although there is experimental evidence for the existence of electromagnetic mass, the conventional theory ``falls  on its face" in accounting for this mass ..., ``because it does not produce a consistent theory--and the same is true for quantum modifications". 

The last term in equation (18) is an addition to the Lorentz force, with the opposite sign of $-\nabla \Phi $, which appears in $\bf{E}$.  In order to see the physical meaning of the term, assume an interaction between a proton and an electron, where $\bf{A}=0$ and $V$ is the Coulomb interaction, so that (17) and (18) become:
\beqa 
\begin{gathered}
  {\mathbf{u}} = \left[ {1 + \frac{V}
{{mc^2 }}} \right]\frac{\pi }
{m}, \hfill \\
  \frac{c}
{b}\frac{{d{\mathbf{p}}}}
{{d\tau }} =  - \nabla V - \nabla V\frac{V}
{{mcb }}. \hfill \\ 
\end{gathered} 
\eeqa
If we treat $\bf{u}$ as (approximately) ${\bf{p}}/m$ and set $b=c$, we get that
\[
m{\bf{a}} =  - \nabla V - \nabla V\frac{V}
{{mc^2 }}.
\]
In this case, it is easy to show that the classical electron radius, $r_0$, is a critical point (i.e., $-\nabla \Phi   -\nabla \Phi (V/mc^2)=0$).  Thus, for $0< r< r_0$, the force becomes repulsive.  We interpret this as a fixed region of repulsion, so that the singularity $r=0$ is impossible to reach at the classical level.  The neglected terms are attractive but of lower order.  This makes the critical point slightly less than $r_0$.  Thus, in general, the electron experiences a strongly repulsive force when it gets too close to the proton.  This means that the classical principle of impenetrability, namely that no two particles can occupy the same space at the same time occurs naturally. Furthermore, this analysis shows conclusively that  information about the (classical) structure of the particle is not required in the canonical proper-time theory.  

Finally, It is clear that the neglect of second order terms gives us the non-relativistic theory. 
\section{\bf Many-Particle Case}
Once it was agreed that the ``correct formulation of" relativistic classical mechanics should be invariant under the Lorentz group, work on this problem was generally ignored until after World War Two when it was realized that quantum theory did not solve the open problems of classical electrodynamics. In particular, it was first noticed that the canonical center-of-mass is not the three-vector part of a four-vector (see Pryce [21]).  The well-known no-interaction theorem of Currie, Jordan and Sudarshan [22] shows that it is impossible to construct a (interacting) relativistic many-particle theory that allows covariance and independent particle world-lines. (For a discussion of this and all known problems, see  [18].) Thus, the four-vector approach ``falls on its face" for more than one particle.  In this section we construct a consistent classical (relativistic) many-particle theory that is quantizable and includes  Newtonian mechanics.

Suppose we have a closed system of $n$ interacting particles with Hamiltonians $H_i $ and total Hamiltonian $H$.  We assume that $H$ is of the form $H = \sum\limits_{i = 1}^n {H_i }$.  If we define the effective mass $M$ and total momentum $ {\bf P}$ by
\[ 
Mc^2  = \sqrt {H^2  - c^2 {\bf P}^2 },\quad {\bf P} = \sum\limits_{i = 1}^n {{\bf
p}_i },  
\]          
then $H$ also has the representation
$H = \sqrt {c^2 {\bf P}^2  + M^2 c^4 }$. To construct the many-particle theory, we observe that the representation $ d\tau  = (Mc^2 /H)dt$
does not depend on the number of particles in the system and is an invariant for all observers (see [18]).  Thus, we can uniquely define the  proper-time of the system for all observers.    If we let ${\bf L}$  be the boost (generator of pure Lorentz transformations) and define the total angular momentum ${\bf J}$ by   
$$ 
{\bf J} = \sum\limits_{i = 1}^n {{\bf x}_i  \times {\bf p}_i } ,   
$$  
we then have the following Poisson Bracket relations characteristic of the Lie algebra for
the Poincar\' e group (when we use  the observer proper-time):   
$$ 
{d{\bf  P}
\over{dt} }= \left\{ {H,{\bf P}} \right\} = {\bf 0}  \qquad  {d{\bf J} \over {dt}} =
\left\{ {H,{\bf J}} \right\} = {\bf 0}  \qquad  \left\{ {P_i ,P_j } \right\} = 0   
$$  
$$ 
{
\left\{ {J_i ,P_j } \right\} = \varepsilon _{ijk} P_k  }  \qquad  { \left\{ {J_i ,J_j }
\right\} = \varepsilon _{ijk} J_k } \qquad  { \left\{ {J_i ,L_j } \right\} = \varepsilon
_{ijk} L_k }    § 
$$  
$$ 
{ {d{\bf  L}\over dt}
 = {\left\{ H,{\bf L} \right\} } = - {\bf P}  }  \qquad  { \left\{ {P_i , L_j } \right\} =  -{
\delta _{ij}} H/c^2, } \qquad  { \left\{ {L_i ,L_j } \right\} =  - {\varepsilon _{ijk}}
J_k /c^2. }  
$$  
It is easy to see that $M$ commutes with $H$, ${\bf P}$,  and ${\bf J}$, and to show
that $M$ commutes with ${\bf L}$. Constructing $K$ as in the one-particle case, we
have  
$$ 
K = {{H^2 } \over {2Mc^2 }} + {{Mc^2 } \over 2} = {{{\bf P}^2 } \over
{2M}} + Mc^2 . 
$$					 
Thus, we can use the same definitions for ${\bf P}$, ${\bf J}$, and ${\bf L}$  to
obtain our new commutation relations: 
$$ 
{{d{\bf P}} \over {d\tau }} = \left\{ {K,{\bf P}} \right\} = {\bf 0}, \quad {{d{\bf
J}} \over {d\tau }} = \left\{ {K,{\bf J}} \right\} = {\bf 0}, \quad \left\{ {P_i ,P_j }
\right\} = 0, 
$$ 
$$ 
\left\{ {J_i ,P_j } \right\} = \varepsilon _{ijk} P_k, 
\quad
\left\{ {J_i ,J_j } \right\} = \varepsilon _{ijk} J_k , \quad \left\{ {J_i ,L_j } \right\} =
\varepsilon _{ijk} L_k,  
$$ 
$$ 
{{d{\bf L}} \over {d\tau }} = \left\{
{K,{\bf L}} \right\} = {{ - H} \over {Mc^2 }}{\bf P}, \quad \left\{ {P_i ,L_j } \right\}
=  - \delta _{ij} H/c^2,  \quad \left\{ {L_i ,L_j } \right\} =  - \varepsilon _{ijk} J_k
{\kern 1pt} /{\kern 1pt} c^2.   
$$			 
It follows that, except for a constant
scale change, the inhomogeneous proper-time group is generated by the same Lie algebra as the  Poincar\'{e} group. This result is not surprising given the close relation between the two groups.  It
also proves our earlier  statement that the form of $K$ is fully relativistic.

Let the map from $( {\bf x}_i, t)\, \to \,({\bf x}_i, \tau)$ be denoted by ${\bf
C}[\,t,\,\tau]$, and let  ${\bf P}(O', O)$ be the Poincar\' e map from $O \to O'$.
\begin{thm}  The proper-time coordinates of the system
as seen by an observer at $O$ are related to those of an observer  at $O'$ by the
transformation: 
$$ 
{\bf R}_M[\tau ]={\bf C}[ \,t',\,\tau ]{\bf P}(O',\,O){\bf C}^{-1}[ \,t,\,\tau ]. 
$$
\end{thm}
\begin{proof}The proof follows since the diagram below is
commutative.   
\[
 \begin{matrix} {O(\{{\bf x}_i \},\, t)}  &  {\rm  {}} & { \longrightarrow} 
&  {\rm  {}} & O'(\{{\bf x'}_i \},\, t') \cr &  {\rm  {}}&  {\rm  {}}&  {\rm  {}}& 
{\rm  {}}&  {\rm  {}}\cr &  {\rm  {}}&  {\rm  {}}&  {\rm  {}}&  {\rm  {}}&  {\rm 
{}}\cr {{{\bf C}^{-1}[\,t,\,\tau]}} &   \Bigg\uparrow &   {\rm  {}} & 
\Bigg\downarrow & {{\bf C}[\,t',\,\tau]}\cr &  {\rm  {}}&  {\rm  {}}&  {\rm  {}}& 
{\rm  {}}&  {\rm  {}}\cr &  {\rm  {}}&  {\rm  {}}&  {\rm  {}}&  {\rm  {}}&  {\rm 
{}}\cr O(\{{\bf x}_i \},\, \tau)  &  {\rm  {}} & \longleftarrow &  {\rm  {}} &
O'(\{{\bf x'}_i \},\, \tau)
\end{matrix} 
\]
\end{proof}
The top diagram is the Poincar\'{e} map from $O \to O'$.  {\sl It is important to note that
this map is between  the coordinates of observers}.  In this sense, our approach may be
viewed as a direct generalization of the conventional theory.  In the global case, when
${\bf U}$ is constant, $t$ is related to $\tau$  by a scale transformation so that we
have a group with the same Lie algebra as the Poincar\'{e} group (up to a constant scale),
but it has an Euclidean metric.  In this case, Theorem 2 proves that ${\bf R}_M$ is
in the proper-time group formed by a similarity action on  the Poincar\'{e} group by the
canonical group ${\bf C}_\tau$.    On the other hand, Theorem 2 is true in general.
This means that in both the local and global cases (when the acceleration is nonzero)
$t$ is related to $\tau_i$ and $\tau$ via nonlocal (nonlinear) transformations.  It
follows that, in general, the group action is not linear, and hence is not covered by the
Cartan classification.  
 
Since $K$ does not depend on the center-of-mass position ${\bf X}$, it is easy to see
that  
\beqn 
{\bf U} = {{d{\bf X}} \over {d\tau }} = {{\partial K} \over {\partial {\bf
P}}} = {{\bf P} \over M} =  {1 \over M}\sum\limits_{i = 1}^n {m_i {\bf u}_i }, 
\eeqn 
where ${\bf u}_i  = d{\bf x}_i /d\tau _i $.  We can now define $b$ by 
\[ 
b = \sqrt {{\bf U}^2  + c^2 }  \Rightarrow H = Mcb. 
\]
Thus, we can represent the relationship between $d\tau$ and $dt$ as: 
\[ 
d\tau  = (c/b)dt.
\]
If we set ${\bf v}_i  = d{\bf x}_i /d\tau $, an easy calculation shows that   
\[
{\bf{u}}_i  = {{d{\bf x}_i } \over {d\tau _i }} = {{d\tau } \over {d\tau _i }}{{d{\bf x}_i }
\over {d\tau }} =  {{b_i } \over b}{\bf v}_i  \Rightarrow {{{\bf u}_i } \over {b_i }} =
{{{\bf v}_i } \over b}.  
\]
The velocity ${\bf v}_i $ is the one our observer sees when he uses the global
proper-clock of the system to compute  the particle velocity, while ${\bf u}_i $ is the
one seen when he uses the local proper clock of the particle to compute its velocity. 
Solving for ${\bf u}_i $ and $b_i $ in terms of ${\bf v}_i $ and $b$, we get  
\[ 
{\bf{u}}_i  = {{c{\bf v}_i } \over {\sqrt {b^2  - \mathop {\bf v}\nolimits_i^2 }}}, {\rm  
}b_i  =  {{cb} \over {\sqrt {b^2  - \mathop {\bf v}\nolimits_i^2 }}} \,\,\,  {\rm or }
\,\,\,  {{b_i } \over b} = {c \over {\sqrt {b^2  - \mathop {\bf v}\nolimits_i^2 } }}.  
\]
Note that, since $b^2  = {\bf U}^2  + c^2$, if ${\bf U}$ is not zero, then any ${\bf
v}_i $ can be larger than $c$.   On the other hand, if  ${\bf U}$ is zero, $b = c$ and,
from the global perspective, our theory looks like the  conventional one. 
Using (19),
we can rewrite ${\bf U}$ as  
\[ 
{\bf U} = {1 \over M}\sum\limits_{i = 1}^n {m_i
{\bf u}_i }  =  {1 \over M}\sum\limits_{i = 1}^n {{{m_i c{\bf v}_i } \over {\sqrt
{b^2  - \mathop {\bf v}\nolimits_i^2 } }}}  = {1 \over M}\sum\limits_{i = 1}^n
{{{b_i m_i {\bf v}_i } \over b}}  = {1 \over H}\sum\limits_{i = 1}^n {H_i {\bf v}_i
}. 
\] 
It follows that the position of the center-of-mass (energy) satisfies   
\[
{\bf X} = {1 \over H}\sum\limits_{i = 1}^n {H_i {\bf x}_i }  + {\bf Y},
\quad {d{\bf Y} \over d{\tau}}={\bf 0}.   
\]
It is natural to choose ${\bf Y}$ so that ${\bf X}$ is the canonical center of mass:   
\[
{\bf X} = {1 \over H}\sum\limits_{i = 1}^n {H_i {\bf x}_i }  +  {{c^2 ({\bf S} \times
{\bf P})} \over {H(Mc^2  + H)}}, 
\]		              
where  ${\bf S}$ is the (conserved) spin of the system (see [23]).  The important point is that 
${\rm (}{\bf X},{\bf P},\tau, K{\rm )}$ is the new set of (global) variables for the
system. 

As the system is closed, ${\bf U}$ is constant and $\tau $ is linearly related to $t$.
Yet, the physical interpretation is very different  if ${\bf U}$ is not zero. 
Furthermore, it is easy to see that, even if ${\bf U}$ is zero in one frame,
it will not be zero in any other frame which is in relative motion (see the next section).  It is clear that $\tau $
is uniquely determined by the particles in the system and is available to all observers.  Thus, it offers a unique definition of simultaneity for all events associated with the global system.  

In general, an individual particle in a large system may not be directly observable (e.g., a small planet near a large sun in a distant galaxy). On the other hand, if an individual particle is observable, we have another unique definition of simultaneity for all events associated with that particle.  Furthermore, if a  subsystem of particles is observable, the local proper clock of the subsystem offers yet another unique definition of simultaneity (for all events associated with it).  Thus, the convention used provides a unique definition of simultaneity. 

It should be noted that there is a basic relationship between the global
system clock and the clocks of the individual particles.  In order to derive this relationship, we return to our definition of the global Hamiltonian $K$ and let $W$ be
any observable.  Then  
\begin{align}
  & {{dW} \over {d\tau }} = \left\{ {K,W} \right\} = {H \over {Mc^2 }}\left\{ {H,W}
\right\} =  {H \over {Mc^2 }}\sum\limits_{i = 1}^n {\left\{ {H_i ,W} \right\}}   \cr 
  & \qquad = {H \over {Mc^2 }}\sum\limits_{i = 1}^n {{{m_i c^2 } \over {H_i
}}\left[ {{{H_i } \over {m_i c^2 } }\left\{ {H_i ,W} \right\}} \right]}  =
\sum\limits_{i = 1}^n {{{Hm_i } \over {MH_i }}\left\{ {K_i ,W} \right\}}.  
\end{align} 
Using the (easily derived) fact that $d\tau _i /d\tau  = Hm_i
/MH_i  = b_i /b$, we get 		  
\beqn 
{{dW} \over {d\tau }} = \sum\limits_{i = 1}^n
{{{d\tau _i } \over {d\tau }}\left\{ {K_i ,W} \right\}}. 
\eeqn           
Equation (21) allows us to relate the global systems dynamics to the local systems dynamics and  provides the basis for a direct approach to the quantum relativistic many-body problem using one (universal) wave function.  

We now show directly that the transformation, at the global level, is a canonical change of variables (time). 
\begin{thm}
There exists a function $S = S\left( {\{ {\bf{x}}_i \} ,\;\{ {\bf{p}}_i \} ,\;\tau } \right)$ such that
\[
\begin{gathered}
{\mathbf{P}} \cdot d{\mathbf{X}} - Hdt \equiv {\mathbf{P}} \cdot d{\mathbf{X}} - Kd\tau  + dS, \hfill \\
 \sum _{i = 1}^n {\mathbf{p}}_i  \cdot d{\mathbf{x}}_i  - \sum _{i = 1}^n H_i dt \equiv \sum _{i = 1}^n {\mathbf{p}}_i  \cdot d{\mathbf{x}}_i  - Kd\tau  + dS. \hfill \\ 
\end{gathered} 
\]
\end{thm}
\begin{proof}
Set $S = [Mc^2  - K]\tau$.   An easy calculation, using the fact that both $Mc^2$ and $K$ are conserved quantities, shows that $dS=[Mc^2  - K]d{\tau}$.  An additional easy calculation gives the  result.
\end{proof}
It should be observed that, in a manner similar to that of Horwitz and Piron [24], we can formulate a dynamical principle which generalizes Hamilton's principle by using the integral invariant of Poincar\'{e}-Cartan (see Arnol'd [25]):  
\beqa
I = \oint_{\bf{C}} {\sum\nolimits_{i = 1}^n {{\bf{p}}_i }  \cdot d{\bf{x}}_i  - } Kd\tau, 
\eeqa
where $\bf C$ is a closed curve on extended phase space $\G= \G \left( {\{ {\bf{x}}_i \} ,\;\{ {\bf{p}}_i \} ,\;\tau } \right)$, and the above integral is invariant for arbitrary deformations of $\bf C$ along trajectories corresponding to solutions of the equations of motion.

A fundamental conclusion of this section is that, for any system of particles, we can always choose a unique observer-independent measure of time that is intrinsically  related to the local clocks of the individual particles.  (This is not true for any of the other attempted formulations of either a classical or quantum relativistic many-particle theory.) One important consequence of this result can be stated as a theorem.
\begin{thm}
Suppose that the observable universe is representable in the sense that the observed ratio of mass to total energy is constant and independent of our observed portion of the universe.  Then the universe has a unique clock that is available to all observers.
\end{thm}
The above assumption is equivalent to the  homogeneity and isotropy of the energy and mass density of the universe. 

In the study of physical systems one is sometimes not interested in the behavior of the global system, but only in some subsystem.  The cluster decomposition property is a requirement of any theory purporting to be a possible representation of the real world.  Basically this is the property that, if any two or more subsystems become widely separated, then they may be treated as independent systems (clusters).  
\begin{thm}
Suppose that our system of particles can be decomposed (by the observer) into two or more clusters.  Then there exists a unique (local) clock and corresponding canonical Hamiltonian for each cluster. 
\end{thm}
\begin{proof}
We assume that the subsystems are sufficiently separated that all observers can agree that they are distinct.  In this case, each observer can identify effective masses $M_1, \; M_2$ and Hamiltonians $H_1, \; H_2$.  It follows that $d\tau_1=[(M_1 c^2)/H_1]dt$ and $d\tau_2=[(M_2 c^2)/H_2]dt$, so that each observer can construct a local-time theory for each cluster.
\end{proof}
Actually, the theorem is true without the assumption that the systems are weakly interacting.  This makes the theorem less difficult to apply than the various phenomenological approaches, which require both model justification and consistency analysis prior to use.  This theorem also allows us to prove a weaker version of Theorem 4, in the sense that we replace the assumption of  homogeneity and isotropy of the energy and mass density for a possible infinite universe by finite energy and mass density for a possibly inhomogeneous universe.
\begin{thm}
Suppose the universe has finite mass and energy density and that each observer can choose a local inertial frame from which his/her region of the universe is at rest relative to the observed system.  Then there exists a unique proper clock for the universe. 
\end{thm}
\begin{proof}
Applying the cluster decomposition theorem, our observer can identify masses $M_1$ for his/her region of the universe and $M_2$ for the complement region, along with Hamiltonians $H_1$ and $H_2$.  It follows that $H=H_1+ H_2,\; M=M_1+ M_2$ and $d\tau=[(M c^2)/H]dt$ define the total mass, Hamiltonian and proper clock for the universe.  We can now construct our canonical proper-time Hamiltonian $K$.  Since $M$ and $H$ are fixed, and invariant for all observers, we see that both $K$ and $\tau$ are unique and invariant for all observers. 
\end{proof}
\begin{rem}
It should be remarked that $M_i$ and $H_i,\; i=1,2$, will vary with observers, reflecting the nonuniqueness of inertial frames.
\end{rem}
\subsection{Preferred Rest Frame}
It has been known since the pioneering work of Penzias and
Wilson [26]  that a unique preferred frame of rest exists
throughout the universe and is available to all observers.  This is the 2.7 $^{\circ}$K microwave background radiation (MBR) which was discovered  in 1965 using basic microwave equipment (by today's standards).  This radiation  is now known to be highly isotropic with anisotropy limits  set at $0.001\%$.    Futhermore, direct measurements have been made of the velocity of
both our Solar System and Galaxy through  this radiation (370 and 600 $km/sec$ respectively, see Peebles [27] ).  One can only speculate as to what impact this information would have had on the thinking of Einstein, Lorentz, Minkowski, Poincar\'e, Ritz and the many other investigators of the early 1900's who were concerned with the  foundations of electrodynamics and mechanics.  The importance
of this discovery for the foundations of electrodynamics is that this frame is determined by radiation from all accelerated charged particles in the universe.  

This frame has not found a natural place in the standard framework for  theoretical physics and in fact, has been (almost) ignored. (However, Glashow and co-workers have used it as a part of a program to explore possible departures from strict Lorentz invariance in the context of elementary-particle kinematics (see Cohen and Glashow [28], and cited references)).   

Since all inertial reference frames are equivalent, the one chosen by any observer is a convention.  If we seek simplicity in this basic representation for physical reality, it is natural to attach all (globally defined) frames to the MBR, and use the proper-time of the universe for the clock.  In this case, the two postulates of the special theory are automatically satisfied, while the field and particle equations of any system will be invariant (not covariant) under the action of the Lorentz group at the global level (for all observers).  Furthermore, the speed of light would  always be $c$ relative to this frame.  Thus, we suggest this as the natural choice for all global systems. With this choice, we can have $c$ unique at the global level, while speeds larger then $c$ are allowed and not unnatural at the local level.  This approach has the additional advantage of eliminating all the known paradoxes associated with the special theory.
\subsection{Time Reversal Noninvariance}
We focus on the single particle case. (The same discussion applies to the many-particle case.) Since $d\tau=(mc^2/H)dt$, while $K(=[H^2/2mc^2+mc^2/2])$ and $m$ are always positive, we see that, if $H \rightarrow -H$ (negative energy) or $t\rightarrow-t$ (time reversal), then $K\rightarrow K$ is invariant (no negative energy), while $\tau$ changes sign. {\it In particular, our theory is noninvariant under observer time reversal at the classical level and, since $\tau$ is monotonically increasing, we acquire an arrow for (proper) time.} It is thus natural to interpret antimatter as matter with its proper-time reversed.
 
A more complete (and elegant) discussion requires the introduction of Santilli's isodual numbers [29], in which the unit 1 is replaced by $-1$ and $ab \rightarrow a \bullet b=-ab$ so that $(-1)\bullet(-1)=-1$.  It follows that $\left( {\mathbb{R}, + , \bullet } \right)$ is a field.    This clearly induces an isomorphism on $\R$ which is equivalent to reversing the direction of the real line. Santilli shows that use of the isodual numbers as the scalar field on the dual Hilbert space $\mathcal{H}^*$, allows him to provide a consistent formulation of the Stuckelberg-Feynman (quantum) theory of antimatter as matter in the isodual state (over $\mathcal{H}^*$), and proved that this gives an equivalence between charge conjugation and isoduality.

Thus, by introducing a symmetric theory of numbers, we can construct a completely symmetric theory of matter which avoids all of the natural objections to hole theory, while maintaining consistency with our physical sense of a monotonically increasing time variable.  Both Feynman [30] and Stueckelberg [31] introduced the idea of representing antimatter as matter with its time reversed.  
Our final conclusion is the same as theirs.  However, the two approaches are distinct.  In our approach, we replace $t$ by $\tau$ and acquire $K$ as its canonical Hamiltonian, so all physical interpretations only require information about $\tau$.   

The quantum theory now follows by replacing the Poisson bracket in equation (21) by the Heisenberg bracket, which leads to Schr\"{o}dinger-like equations:
\[
i\hbar \tfrac{{\partial \psi }}{{\partial \tau }} = K\psi ,\;{\text{and }}i\hbar \tfrac{{\partial \psi }}
{{\partial \tau _i }} = K_i \psi, 
\]
for the same (universal) wave function $\psi=\psi(\bf{X}, \tau)=\psi({\bf{x}}_1, \cdots,{\bf x}_n,\tau_1, \cdots, \tau_n)$ (see equation (21) and Theorem 3).   Since $K,\;K_i$  are both positive definite, as operators, they are bounded below.   Thus, the problem of negative energy which caused confusion during the early attempts to merge quantum mechanics and the special theory of relativity do not arise. (From equations (20) and (21), we see that $K$ is not equal to the sum of the $K_i$.)

The question of particle number is easily included (even in the classical case) by observing that, for any closed system of interacting particles, we can replace the definite particle number $n$ by a variable (random) particle number $N(t)$, the number of particles up to time $t$ (as seen by the observer).   Now, the only relevant variables are the conserved quantities: the total global energy, momentum, angular momentum and spin and, as in QED, for large negative $t,\;N(t) \to n_i $, (the initial particle number), and for large positive $t,\;N(t) \to n_f $,  (the final particle number), where $n_i $ and $n_f $ are assumed known from experiment.

\section{Light: Its Nature, Its Mass and Its Speed}
\subsection{\bf Photon Nature} 
Given that, in our local-time formulation of the Maxwell theory, the value of the ``natural" speed of light depends on the motion of the source, it is not presumptuous to take another look at the age-old problem of its nature and seek to understand the impediments to treating a photon on the same footing as any other elementary particle.

The single most important problem with any attempt to treat photons as elementary particles is its well-known wave property.   Interference and diffraction experiments and, indeed, a number of major fields of electrical engineering attest to the amazing precision, effectiveness and efficiency of the wave picture.  The field theory view of both photons and particles is that they are localized wave packets of the quantized electromagnetic field (or the particular particle field).  Hence, one expects that matter will behave like particles in some experiments and like waves in others and, as such, may be viewed as some type of atomic compromise between two distinct classical views.  However, these ideas go back to the fundamental work of Born, deBroglie and Heisenberg, while today we have much more experimental information about photons since those pioneering efforts (see the next section). 

It has now been known for over twenty-five years that we can control the intensity of a beam of photons in interference (or diffraction) experiments to the point that individual photons may be counted on a photographic plate (see Paul [32]).  Furthermore, the distribution of photons appears random and, only after a long time period (depending on the intensity level), do we begin to see wave patterns.  These experiments, along with the photoelectric and Compton effects, conclusively tell us the following:
\begin{enumerate}
\item It is not unreasonable to treat photons as elementary particles.
\item The wave length and frequency are characteristic of groups of photons. 
\item Concepts such as electric and magnetic fields are macro-properties that have some (but limited) reality on the atomic scale and (possibly) none at the sub-atomic scale.
\end{enumerate}
For an excellent discussion of the above, see Buenker [33].     
\subsection{\bf Photon Mass}  
In the past, work on the question of photon mass has focused on the addition of a mass term to  the Lagrangian density for Maxwell's equations and generally leads to the Proca equation  (see Bargmann and Wigner [34]).  Early work in this direction can be traced back from the  paper of  Schr{\"o}dinger and Bass [35].  As in our
approach, the speed of light is no longer constant in all reference frames.  In this case, the fields are distorted by the mass term and  experiments of Goldhaber and Nieto [36] use geomagnetic data to set an upper bound of  $3\times 10^{-24}\,GeV$ for the mass term (see Jackiw [37]).  This approach causes gauge formulation difficulties, and  has not found favor at the classical level.   

It should be recalled that Maxwell's equations are (spin $1$) relativistic wave equations (see Akhiezer and Berestetskii [39]). On the other hand, experiments show directly that photons have a weight as one would expect of any material object.  These experiments do not depend on either the special or general theory of relativity (according to Pound and Snider [40]) and are not directly dependent on frequency or wavelength measurements.

The results of this paper establish that, at the classical level, photons are created by external forces acting on charged particles.  These  photons  arise due to the resistance of charged particles to any change in their motion.  From our derivation of the wave equations, it is clear that the theory is fully gauge invariant.  Furthermore, from equation (6), we see that the effective (photon) mass is dynamical, appearing only during acceleration of the source.  From equation (7), we see that this mass generates the expected longitudinal component to the electromagnetic field. 

In a recent paper, Afshar et. al. [58] report on the presence of sharp interference and highly reliable which-way information in the same experimental arrangement for the same photons using new 
non-perturbative measurement techniques at separate spacetime coordinates, both of which refer back to the behavior of the photon at the same event (i.e., the passage through the pinholes).  

They inferred full fringe visibility from the observation that the total
photon flux was only slightly decreased when thin wires were placed
exactly at the minima of the interference pattern. Which-way
information was then obtained downstream via  imaging using a special lens system.  This allowed them to determine which-way information about each photon as it passed the plane of the pinholes.   The experiment allowed them to circumvent the predicted limitations imposed by the uncertainty principle and the entanglement between the which-way marker and the interfering quantum object as employed in other experiments.  They further point out that the  non-perturbative measurement technique used in the experiment is conceptually quite different from quantum non-demolition or non-destructive techniques which perturb the photon wavefunction directly. This shows that wave-particle duality is more a property of experimental approach and technique as opposed to an intrinsic property of photons.
  
It is not unreasonable to correlate the frequency of a photon with that of linear harmonic oscillations for low energy fields.  Drozov and Stahlhofen [59] note that the Planck formula, ${\mathcal E}=h \nu$, was based on the fact that the energy emitted from atoms resulted from electromagnetic waves  propagating as continuous harmonic oscillations. In order to solve the cavity problem, Planck postulated energy quantization in terms of the frequencies defined by the oscillations, leading to the above formula, where the total energy of the emitted electromagnetic wave is an integer multiple of $h$.  From the success of this approach, it is now natural to believe that photons are small parts of the electromagnetic field. However, such a view is untenable for describing the ultra-short pulses that are obtained in laboratories today.  If we combine the Planck formula with $c=\la \nu$, one must then accept that $\la = hc/{\mathcal{E}}$.   Since we now don't have wave-particle dualism to save us, the concept of a photon as a wave begs the question posed by  Drozov and Stahlhofen: ``how many oscillations does it contain?".  We agree with them that it is much more reasonable to relate $\nu$ to the source emission time.  From this view, the emission time is only restricted by the time energy 
uncertainty relationship $\De {\mathcal{E}} \tau \approx h/2$.  Thus, we conclude that there is no oscillation period with frequency $\nu$, just the characteristic time $\tau =1/(2\nu)$, for a single (massive) energy pulse.  (We note that the use of a Fourier series representation, as is generally done in optics, is no more than a good mathematical device and should not be used for physical interpretation.)  In closing, we point out that, as noted by Feynman [38], a small photon mass will eliminate the infrared catastrophe in QED.   
     
\subsection{Light Speed}
In this and the remaining sections, we  give additional consideration to the physical implications of our interpretation of $b$ and $b'$ as the speed of light relative to the source for the different observers (collaborative speed of light).  We also consider a few topics in astrophysics in which the conventional theory requires additional assumptions in order to explain observations that point to speeds higher than $c$. 

In order to gain some perspective, suppose an emitting system is at rest in the unprimed frame so that $b=c$. In this case, the collaborative speed of light observed in the primed frame is $b'=\gamma( {\bf{v}})c$ and the velocity of the source is seen as ${\bf{u}}'=-\gamma( {\bf{v}}){\bf{v}}$. Thus, if the two observers are separating at high speeds, both $b'$ and ${\bf{u}}'$ may be very large.     

There are some experiments where use of the observer's clock provides a clear answer.  A classic example is the Michelson-Morley experiment.  This experiment gave the first bell of doom for the ether theory, and is easily explained by the special theory (using the clock of the observer). It also has a simple explanation when the clock of the source is used since, in this case, the source is at rest in the frame of the observer so that ${\bf{u}} = {\bf{0}} \Rightarrow b=c$. 

It is clear that, at the speeds obtained in the world of our ordinary experience, no significant difference between the two approaches is expected.  However, at high energies and/or small distances, we expect differences to show up in a dramatic way.  Indeed they have, but the definition of velocity depends on the clock attached to the observer, ${\bf{w}} = {{d{\bf{x}}} \mathord{\left/ {\vphantom {{d{\bf{x}}} {dt}}} \right. \kern-\nulldelimiterspace} {dt}}$, while all contrary results are interpreted as due to time dilation. For example, without a switch in clocks, the existence of cosmic ray muons at the surface of the earth have no explanation.

 An equally valid interpretation is that the velocity of the system is not {\bf{w}}, but $ {\bf{u}} = {{d{\bf{x}}} \mathord{\left/ {\vphantom {{d{\bf{x}}} {d\tau}}} \right. \kern-\nulldelimiterspace} {d\tau}}$ and, in this case, no contrary results occur.  The use of ${\bf{w}}$ is clearly a convenient choice for most of ordinary physics (where both choices are the same).  However, in high-energy experiments, the local clock of the system is necessary (and used) to analyze and explain scattering events by using time dilation to make the results correspond to velocities below $c$.  
 
In order to obtain a different view of experiments on the lifetime of fast mesons and the velocity of ${\gamma}$ rays and light from moving sources, first consider the definition of momentum.  When the clock of the observer is used to measure time, momentum increase is attributed to relativistic mass increase so
that $ 
{\bf p}=m{\bf w}$ and $ m=m_0[1-w^2/c^2]^{-1/2}.$  
On the other hand, if we use the clock of the source, we have  
$ {\bf p}=m_0{\bf u}$ and ${\bf u}={\bf w}[1-w^2/c^2]^{-1/2}
$  
so that there is no mass increase, the (proper) velocity increases. Thus, in particle experiments, the particle has a fixed mass and invariant decay constant, independent of its velocity, but can have speeds $> c$.  An analysis of experiments on the lifetime of fast mesons, the velocity of ${\gamma}$ rays and light from moving sources reveal that, at some point, either the speed of light is assumed to be independent of the motion of the source, or time dilation is used.  Both assumptions imply that only the clock of the observer is used. This is basically inconsistent since one is using two clocks to explain one experiment.  Either, the speed of the particle is less then $c$ as determined by the observer's measuring instruments and clock, or,  if time dilation is used, one is in fact inferring that the clock of the particle is needed to explain the experiment.  In the latter case, the actual speed of the particle is then $\bf u$, which can be greater than $c$.  The problem is that one cannot use two clocks and claim consistency with the basic measurement framework for special relativity.   {\it Thus, the analysis of these experiments is flawed in its validation of the conventional theory and does not prove that the speed of light relative to the particle is $c$.}  
\subsection{ Relativistic Jets in Our Galaxy}
In 1918 Curtis [41] made the first discovery of jet-like features emanating from the nuclei of galaxies.  He identified a jet in the optical range from an elliptical galaxy in the Virgo cluster (M87).  Since then, a large number of objects with a jet-like structure have been discovered.  However, starting about 30 years ago, researchers began to find quasars with jet expansions of up to ten or more times the speed of light (see Pearson and Zensus [42]; Zensus [43]; Mirabel and Rodr\'{i}guez [44]).   This started quite a stir and led many to suggest the possible breakdown of the special theory.  Since the source appeared to be at rest relative to earth, unlike the decay of fast muons from the top of the atmosphere, the assumption of time-dilation would not work.   However, Rees [45] suggested that we may be looking at the jets from some angle relative to the observer.  He showed how these large speeds may be an aberration because we were looking at a projection of the true image onto the plane of view of the observer.  This would explain the apparent approaching and receding condensations with very different velocities.  

This is an additional assumption with no independent  verification and has not been completely accepted (see Mirabel and Rodr\'{i}guez [44]; De R\'{u}jula [46]).  If our formulation is correct, it is possible that the measured speeds are correct, but the approaching and receding condensations  are caused by different physical mechanisms.  It then follows that the additional assumptions of Rees are not required.      
\subsection{Ultra-high Energy Cosmic Rays}
The nature and origin of ultra-high energy cosmic rays (UHECR) continues to cause controversy and concern.  One year after Penzias and Wilson [26] discovered the cosmic microwave background radiation (CMBR), Greisen [47] and Zatsepin and Kuzmin [48] estimated that the mean free path of an energetic $10^{19} \;eV$ proton moving through the CMBR would be less than the size of our galaxy.  From this work, it was expected that all protons with energies above about  $4 \times 10^{19} \;eV$ (GZK cutoff energy) would be suppressed by dissipative losses in the CMBR.  (This limits how far protons can travel to about 100 Mpc.)  Thus, it was a real surprise when the Akeno Giant Air Shower Array (AGASA) Collaboration [49] reported observations of a large flux of UHECR with energies above $10^{20} \;eV$.  The HiRes (fluorescence detector) Collaboration group [50]  published results that appear consistent with support for a cut-off.  Additional studies are currently being conducted at the Pierre Auger Observatory designed to resolve the discrepancy.   

Since the first draft of this paper, recent results confirm the GZK suppression and appear to indicate that UHECRs above the GZK threshold arrive from nearby active galactic nuclei (see the review by Dar [51].)  

With the conventional definition of velocity, it is very difficult to imagine even the most powerful astrophysical systems, such as active galactic nuclei and/or radio galaxies, accelerating heavy nuclei or protons to the required high energies within existing physical theories.  However, if the local clock of the ejecting system determines the distance traveled and speed ($\bf u$) of a particle, then no new particles or the breakdown of the special theory are needed to explain the results.  Since most measurements (physical and astronomical) are based on the dimensionless ratio ${\bf{\beta}}  = {{\bf{w}} \mathord{\left/ {\vphantom {{\bf{w}} c}} \right. \kern-\nulldelimiterspace} c} \equiv {{\bf{u}} \mathord{\left/{\vphantom {{\bf{u}} b}} \right. \kern-\nulldelimiterspace} b}$, we see that the results will not change.  For example, ([52], pp. 556-561), the red shift factor $z$, used to determined distances in astronomy, is defined by:
\[
z = \sqrt {\frac{{1 + \tfrac{{\mathbf{w}}}
{c}}}
{{1 - \tfrac{{\mathbf{w}} }
{c}}}}  - 1 \equiv \sqrt {\frac{{1 + \tfrac{{\mathbf{u}}}
{b}}}
{{1 - \tfrac{{\mathbf{u}}}
{b}}}}  - 1.
\]
(The formula $z=\tfrac{{\mathbf{w}}}{c} =\tfrac{{\mathbf{u}}}{b}$ is normally used for values of $\left| {\bf w} \right| \le 10^{4} km/sec$.)
We thus conclude that light may reach us from much farther away and with more intensity than is traditionally expected.   Thus, we may well be looking at some galaxies that are not as close and others that are not as far as predicted from conventional theory.

\section{Foundations for Relativistic Quantum Theory}
In this section we provide some discussion of our current work on the corresponding canonical proper-time quantum theory. Because of a number of unexpected problems that have come to our attention, we begin with a discussion of mathematics as it relates to physics. 
\subsection{The Relationship Between Mathematics and Physics} 
With respect to the question of the relationship between physical and mathematical equivalence, we have uncovered three additional instances that shed doubt on such an unanalyzed but important assumption. 
\subsubsection{Feynman Operator Calculus}
In response to the need to provide the mathematical foundations for  Feynman's time-ordered operator calculus used in quantum electrodynamics (QED),  we have developed a constructive theory in which time is accorded its natural role as the director of physical processes (see [53] for a full discussion). 

Briefly, the theory is constructive in that operators acting at different times actually commute. This approach allowed us to develop a general (exact) perturbation theory for all theories generated by unitary groups, by providing the missing remainder term for the Dyson expansion.    We were also able to show that the theory could be reformulated as a physically motivated sum over paths as suggested by Feynman [54].  These results made it possible to prove the last two remaining conjectures of Dyson [55] concerning QED.   
\begin{enumerate}
\item The renormalized perturbation series of QED is at most  asymptotic. 
\item The ultraviolet divergency of QED is caused by a violation of the time-energy uncertainty relations. 
\end{enumerate}
In the Feynman world-view, the universe is a three-dimensional motion picture in which more and more of the future appears as time evolves. Thus, in this view, time is a physically defined variable.  This view is inconsistent with the Minkowski world-view, in which time is a mathematical variable; i.e, an additional coordinate for a space-time geometry.  

In order to see that unitary equivalence cannot always mean physical equivalence, we observe that all solutions to the dynamical Schr\"{o}dinger equation (on the same Hilbert space) generate unitary groups. However, in almost all cases, we would never equate this unitary equivalence with physical equivalence.   

We suggest the following formal definition of physical equivalence.
\begin{Def} Two representations of a given physical system are said to be physically equivalent if they can be related mathematically, and any change of independent variables can be justified by comparison with both experimental as well as theoretical results. 
\end{Def}

From this definition, we conclude that the Lagrange and Hamiltonian formulations of classical mechanics and the Schr\"{o}dinger-Heisenberg formulations of quantum theory are physically equivalent.    On the other hand (see below), the Dirac equation and the square root equation are not physically equivalent.  We also conclude that the canonical proper-time formulation of classical electrodynamics is not physically equivalent to the conventional formulation. 

\subsection{The Dirac Equation and Square-Root Equations}
In order to prepare a solid foundation for the canonical proper-time quantum theory, in [57] we began with an investigation of the Dirac equation.  It is generally believed that it is not possible to separate the particle and antiparticle components directly without approximations (when interactions are present).   We were able to construct an analytical separation (diagonalization) of the full (minimal coupling) Dirac equation into particle and antiparticle components. The diagonalization was analytic in the sense that it was achieved without transforming the wave function or the variables, as is done with the Foldy-Wouthuysen unitary transformation. We started with the Dirac equation in the form:
\[
\begin{gathered}
  i\hbar \frac{{\partial \psi }}
{{\partial t}} = (V + mc^2 )\psi  + c(\sigma  \cdot \pi )\phi  \hfill \\
  i\hbar \frac{{\partial \phi }}
{{\partial t}} = (V - mc^2 )\phi  + c(\sigma  \cdot \pi )\psi.  \hfill \\ 
\end{gathered} 
\] 
Treating these as first order partial differential equations with forcing, we first solved for delta term forcing and obtained the general solution via convolution, to get the following  set of completely separated equations
\[
\begin{gathered}
  i\hbar \frac{{\partial \psi }}
{{\partial t}} = (V + mc^2 )\psi  + \left[ {{{c^2 (\sigma  \cdot \pi )} \mathord{\left/
 {\vphantom {{c^2 (\sigma  \cdot \pi )} {i\hbar }}} \right.
 \kern-\nulldelimiterspace} {i\hbar }}} \right]\int_{ - \infty }^t {\exp \{  - iB(t - \tau )\} (\sigma  \cdot \pi )\psi (\tau )d\tau }  \hfill \\
  i\hbar \frac{{\partial \phi }}
{{\partial t}} = (V - mc^2 )\phi  + \left[ {{{c^2 (\sigma  \cdot \pi )} \mathord{\left/
 {\vphantom {{c^2 (\sigma  \cdot \pi )} {i\hbar }}} \right.
 \kern-\nulldelimiterspace} {i\hbar }}} \right]\int_{ - \infty }^t {\exp \{  - iB'(t - \tau )\} (\sigma  \cdot \pi )\phi (\tau )d\tau },  \hfill \\ 
\end{gathered} 
\]
where $B = \left[ {{{(V - mc^2 )} \mathord{\left/ {\vphantom {{(V - mc^2 )} \hbar }} \right. \kern-\nulldelimiterspace} \hbar }} \right]$ and $B' = \left[ {{{(V + mc^2 )} \mathord{\left/ {\vphantom {{(V + mc^2 )} \hbar }} \right. \kern-\nulldelimiterspace} \hbar }} \right]$.  
We also showed how, for each equation, to construct the complete probability density function:
\[
\begin{gathered}
  \rho _\psi   = \left| \psi  \right|^2  + \left| {\int_{ - \infty }^t {c\exp \{  - iB(t - \tau )\} \left[ {{{(\sigma  \cdot \pi )} \mathord{\left/
 {\vphantom {{(\sigma  \cdot \pi )} {i\hbar }}} \right.
 \kern-\nulldelimiterspace} {i\hbar }}} \right]\psi (\tau )d\tau } } \right|^2  \hfill \\
  \rho _\phi   = \left| \phi  \right|^2  + \left| {\int_{ - \infty }^t {c\exp \{  - iB'(t - \tau )\} \left[ {{{(\sigma  \cdot \pi )} \mathord{\left/
 {\vphantom {{(\sigma  \cdot \pi )} {i\hbar }}} \right.
 \kern-\nulldelimiterspace} {i\hbar }}} \right]\phi (\tau )d\tau } } \right|^2.  \hfill \\ 
\end{gathered} 
\]  
This separation shows that the Dirac equation is actually two equations, both nonlocal in time. It is well-known that the square-root operator is nonlocal in space, but related to the Dirac operator by a Foldy-Wouthuysen unitary transformation.  Thus, the true relationship between the two representations of a spin $1/2$ particle is that one is (explicitly) spatially nonlocal, while the other is implicitly time nonlocal.  We infer that this mathematical equivalence cannot be what we should mean by physical equivalence.

Since our research suggests that time may not be a (mathematical) fourth coordinate, we took another look at the Dirac equation.  In [56], we observed that the operator part of Dirac equation can also be written with the potential energy as a part of the mass ($\Psi =(\psi,\phi)^t$):
\[
D[\Psi ] = \left\{ {c\alpha  \cdot \pi  + \beta mc^2  +  V} \right\}\Psi = \left\{ {c\alpha  \cdot \pi  + \beta \left( {mc^2  + \beta V} \right)} \right\}\Psi. 
\]
From this point of view, we were able to construct another version of the Klein-Gordon equation (assuming that ${\bf A}$ and $V$ do not depend on time):
\beqa
 - \hbar ^2 \frac{{\partial ^2 \Psi }}{{\partial t^2 }} = \left\{ {c^2 \pi ^2  + 2cV\alpha  \cdot {\bf{p}} - e\hbar c\Sigma  \cdot {\bf{B}} - i\hbar c\alpha  \cdot \nabla V + \left( {mc^2  + \beta V} \right)^2 } \right\}\Psi. 
\eeqa
This equation can be factored to give a new square-root equation without any transformation of variables or any change in the wave functions:  
\beqa
i\hbar \frac{{\partial \Psi }}{{\partial t}} = \left\{ {\beta \sqrt {c^2 \pi ^2  + 2cV\alpha  \cdot {\bf{p}} - e\hbar c\Sigma  \cdot {\bf{B}} - i\hbar c\alpha  \cdot \nabla V + \left( {mc^2  + \beta V} \right)^2 } } \right\}\Psi. 
\eeqa
Note that the above equation is not a diagonalized representation.  However, it is an exact representation, which retains the same eigenfunctions and eigenvalues as the Dirac equation.  Thus, we have another analytic representation for the Dirac equation, which does not require a  Foldy-Wouthuysen unitary transformation.  Also observe that, when $V=0$, we get the standard square-root operator equation without potential.   We conclude that the square-root operator equation with potential: 
\beqn
i\hbar \frac{{\partial \Psi }}{{\partial t}} = \left\{ {\beta \sqrt {c^2 \pi ^2  - e\hbar c\Sigma  \cdot {\bf{B}} + m^2c^4 }} +V\right\}\Psi, 
\eeqn
does not represent the same physics as the Dirac equation. 

As noted above, it is known that the square-root operator is spatially nonlocal but, to our knowledge, there was no information in the literature about the nature of the nonlocality or its relationship to actual physical particles.  In [56], we  constructed an analytic representation of the square-root energy operator, which is valid for all values of the spin: ($S[\psi ]({\mathbf{x}}) = \beta \sqrt {c^2 {\mathbf{p}}^2  + m^2 c^4 } \psi ({\mathbf{x}})$)
\[
S[\psi ]({\mathbf{x}}) =  - \tfrac{{\mu ^2 \hbar ^2 c\beta }}
{{\pi ^2 }}\int\limits_{{\mathbf{R}}^3 } {\left[ {\frac{1}
{{\left\| {{\mathbf{x}} - {\mathbf{y}}} \right\|}} - 4\pi \delta \left( {{\mathbf{x}} - {\mathbf{y}}} \right)} \right]\left\{ {\frac{{{\mathbf{K}}_0 \left[ {\mu \left\| {{\mathbf{x}} - {\mathbf{y}}} \right\|\,} \right]}}
{{\,\left\| {{\mathbf{x}} - {\mathbf{y}}} \right\|}} + \frac{{2{\mathbf{K}}_1 \left[ {\mu \left\| {{\mathbf{x}} - {\mathbf{y}}} \right\|\,} \right]}}
{{\mu \,\left\| {{\mathbf{x}} - {\mathbf{y}}} \right\|^2 }}} \right\}\psi ({\mathbf{y}})d{\mathbf{y}}} .
\]  
The functions, ${\bf K}_0, \ {\bf K}_1$ are modified Bessel functions of the third kind.  They have exponential cutoffs, at about a Compton wavelength, like the well-known Yukawa potential ($={\bf K}_{1/2}$). The ${\bf K}_0$ term diverges like $-ln\left|{\bf x}-{\bf y}\right|$, while the ${\bf K}_1$ term diverges like ${\left|{\bf x}-{\bf y}\right|}^{-2}$ at ${\bf x}={\bf y}$.  The delta term acts to cancel the divergency at this point.  Thus, the square-root operator has a representation as a (spatial) nonlocal composite of  three singularities (divergent integrals). In the standard interpretation, the particle component has two negative parts and one (hard core) positive part, while the antiparticle component has two positive parts and one (hard core) negative part. These singularities are confined within a Compton wavelength such that, at the point of divergence, they naturally cancel each other, providing a finite result.    This certainly appears to represent a free three-dimensional soliton with three confined singularities (associated with a linear operator).  

\subsection{Canonical Proper-Time Equations}
We have identified two possible canonical proper-time particle equations for spin-$\tfrac{1}{2}$ particles.  The first and second below are the same, but are derived from two different starting points.  (The second is included in order to validate our assertion that the Dirac equation and square-root equation with potential does not represent the same physics.)  
\begin{enumerate}
\item The canonical proper-time version of the Dirac equation:  
\beqa
i\hbar \frac{{\partial \Psi }}{{\partial \tau}} = \left\{ {\frac{{\pi ^2 }}{{2m}} + \beta V + mc^2  + \frac{{V\alpha  \cdot \pi }}{{mc}} - \frac{{e\hbar \Sigma  \cdot {\bf{B}}}}{{2mc}} - \frac{{i\hbar \alpha  \cdot \nabla V}}{{2mc}}} +\frac{{V^2 }}{{2mc^2 }} \right\}\Psi. 
\eeqa
 \item The canonical proper-time version of the square-root equation, derived from the Dirac equation with the potential energy a part of the mass:
\beqa
i\hbar \frac{{\partial \Psi }}{{\partial \tau}} = \left\{ {\frac{{\pi ^2 }}{{2m}} + \beta V + mc^2  + \frac{{V\alpha  \cdot \pi }}{{mc}} - \frac{{e\hbar \Sigma  \cdot {\bf{B}}}}{{2mc}} - \frac{{i\hbar \alpha  \cdot \nabla V}}{{2mc}}} +\frac{{V^2 }}{{2mc^2 }} \right\}\Psi. 
\eeqa
\item The canonical proper-time version of the standard square-root equation:
\beqn
\begin{gathered}
i\hbar \frac{{\partial \Psi }}{{\partial \tau}} = \left\{ \frac{{\pi ^2 }}{{2m}} - \frac{{e\hbar \Sigma  \cdot {\bf{B}}}}{{2mc }} + mc^2  + \frac{{V^2 }}{{2mc^2 }}\right\}\Psi  \hfill \\
+ \frac{{V\beta \sqrt {c^2 \pi ^2  - ec\hbar \Sigma  \cdot {\bf{B}} + m^2 c^4 } }}{{2mc^2 }}\Psi
+ \frac{{\beta \sqrt {c^2 \pi ^2  - ec\hbar \Sigma  \cdot {\bf{B}} + m^2 c^4 } }}{{2mc^2 }}V \Psi. \hfill \\
\end{gathered} 
\eeqn
\end{enumerate} 
We should note that the eigenvalue problem for the Dirac equation and the corresponding canonical proper-time equation are closely related. In particular, 
\[
E_n \Psi_n  = \left[ {c\alpha  \cdot \pi  + \beta (mc^2  + \beta V)} \right]\Psi_n = \left[ {c\alpha  \cdot \pi  + \beta mc^2  + V} \right]\Psi_n, 
\]
implies that
\beqn
\begin{gathered}
\left[ {\frac{{E_n^2 }}{{2mc^2 }} + \frac{{mc^2 }}{2}} \right]\Psi_n  \hfill \\
 = \left[ {\frac{{\pi ^2 }}{{2m}} + \beta V + mc^2  + \frac{{V\alpha  \cdot \pi }}{{mc}} - \frac{{e\hbar \Sigma  \cdot {\bf{B}}}}{{2m}} - \frac{{i\hbar \alpha  \cdot \nabla V}}{{2mc}}} \right]\Psi_n. \hfill \\
\end{gathered}  
\eeqn
However, the eigenvalue spacing (for the two equations) is clearly not the same.

The objective our study of equations (23) and (24)  is to make sure that the spectrum of Hydrogen cannot be explained as a standard eigenvalue problem before investigating possible many-particle and/or field theory approaches.  We have made substantial progress on the equation (24), and expect to conclude our studies shortly.  

Equation (23) is of paramount interest.  However, our current progress is slowed by the need for additional mathematical research before any important physical implications can be obtained. 
\section*{\bf{Conclusion}}
In this paper, we have provided the foundations for a new approach to relativistic quantum theory.  

First, we have developed a physically and mathematically consistent formulation of classical electrodynamics and mechanics, within which:  
\begin{enumerate}
\item  The invariant speed of light $c$ is replaced by the invariant local clock of the observed system.  In this formulation, the new (collaborative or effective) speed of light is not invariant but depends on the motion of the observed system.
\item Distant simultaneity is unique for all observers.
\item  The local metric is defined in three-dimensional Euclidean space and becomes deformed in the presence of a potential field. 
\item The corresponding canonical Hamiltonian particle theory fixes a particular time direction (i.e., is noninvariant under time reversal).
\item The classical principle of impenetrability, that no two particles can occupy the same space at the same time is intrinsic to the theory.
\item  The formulation of Maxwell's equations is mathematically but not physically equivalent to the conventional one.  
\item This formulation does not depend on the structure of charged particles and does not require self-energy, advanced potentials, mass renormalization, or the problematic Lorentz-Dirac equation in order to account for radiation reaction.
\end{enumerate}
In light of the above, it is no longer clear how far cosmic rays can travel, how far we are from the distant galaxies, or how old the universe is.  Thus, all experiments and observations based on the assumed constant speed of light $c$, relative to all physical systems, needs reevaluation. 

We offer a table below, comparing the two (mathematically equivalent) formulations of classical electrodynamics.

\newpage

\begin{tabular}{r|c||c|c}
 	& {\bf Minkowski} & {\bf Proper-time} \\
\hline \hline
{\bf Reference System} &  inertial required & inertial required \\
\hline
{\bf Light speed} & independent of source & dependent on source \\
\hline
{\bf Space-time} & dependent variables & independent variables \\
\hline
{\bf Transformations} & linear Lorentz & nonlinear Lorentz \\
\hline
{\bf Cluster Property} & highly problematic & general theory \\
\hline
{\bf Many-particle} & highly problematic & general theory \\
\hline
{\bf Radiation reaction} & highly problematic & follows from theory \\
\hline
{\bf Quantum theory} & highly problematic & follows from theory \\
\hline
{\bf Time arrow} & does not exist & follows from theory \\
\hline
{\bf Universal clock} & does not exist & follows from theory \\
\hline
\end{tabular}

\bigskip

We have also discussed our work on a new analytical separation (diagonalization) of the full (minimal coupling) Dirac equation into particle and antiparticle components.  This work reveals the time nonlocal nature of the Dirac equation.  Since it is known that the  square-root equation is spatially nonlocal and, they are related by a unitary transformation, we conclude that this unitary transformation does  not preserve physical equivalence.  

We have introduced a new square-root equation that is obtained from the Dirac equation without any change of variables, by treating the potential energy as a part of the mass.  (This equation has the same eigenvalues and eigenfunctions.)  This approach also leads to a new Klein-Gordon equation.  We have also introduced the quantum version of our canonical proper-time theory.    We obtain two distinct possible wave equations for spin-$\tfrac{1}{2}$ particles.  We are currently investigating these equations in order to rule out the possibility that the spectrum of Hydrogen is obtainable as an eigenvalue problem.   

\bigskip
\noindent{\bf Acknowledments}  

We would like to thank Professor L. Horwitz from the University of Tel Aviv, Israel  and  Professor H. Winful of the University of
Michigan  for their continued interest, encouragement and constructive suggestions.  We would also like to thank the referee for a number of helpful insights. 

\bigskip

\noindent {\bf References}

\smallskip\noindent [1]  A. Einstein, {\it Ann. Phys.} (Leipzig) {\bf{17}}(1905) 891.

\smallskip\noindent [2] P. W. Bridgman, {\it A Sophisticate's Primer of Relativity},  Wesleyan University Press, Middletown Conn. (1962).

\smallskip\noindent [3] E. M. Purcell, {\it Electricity and Magnetism}, Vol. II  McGraw-Hill, New York. (1984).

\smallskip\noindent [4] A. I. Miller, {\it Frontiers of Physics 1900-1911},  Birh\"{a}user, Boston Mass. (1986).

\smallskip\noindent [5]  A. Einstein, {\it Ann. Phys.} (Leipzig) {\bf{17}}(1905) 919.

\smallskip\noindent [6]  M. Planck, {\it Physikalische Zeitschrift}  {\bf{10}}(1909) 825.

\smallskip\noindent [7]  L. V. Lorentz, {\it Philosophical Magazine}  {\bf{34}}(1867) 287.

\smallskip\noindent [8] N. Hamdan, A.K. Hariri and J. L\'{o}pez-Bonilla, {\it Hadronic Journal}  {\bf{30}} (2007) 513. 

\smallskip\noindent [9] W.  Ritz, Ann. Chim. Phys. {\bf 13}, 145 (1908).

\smallskip\noindent [10]  A. Einstein, {\it Physikalische Zeitschrift}  {\bf{10}}(1909) 821.

\smallskip\noindent [11] Autobiographical Notes {\it Albert Einstein: Philosopher-Scientist} vol 1, ed P. A. Schilpp,  Open Court Press, Peru IL (1969).

\smallskip\noindent [12]  H. R. Brown, {\it Eur. J. Phys.}  {\bf{26}}(2005) S85.

\smallskip\noindent [13]  J.A. Wheeler, R.P. Feynman, {\it Rev. Mod. Phys.} {\bf{21}} {(1949) 425.}

\smallskip\noindent [14] P. A. M. Dirac, {\it Proc. R. Soc. London} {\bf A167}, 148 (1938).

\smallskip\noindent [15]  T. Damour, {\it Ann. Phys.} (Leipzig) {\bf{17}} N. 8(2008). (It is not yet in print but is online in electronic form.)

\smallskip\noindent [16] W. Perret and G. B. Jeffery (translators, with additional notes by A. Sommerfeld),   {\it The Principle of Relativity} by H. A. Lorentz, A. Einstein, H. Minkowski and H. Weyl, Dover, New York (1952).

\smallskip\noindent [17] S. Walters, in  {\it The Expanding Worlds of General Relativity}, H. Gonner, J. Renn and J. Ritter (eds.) (Einstein Studies, vol. 7) pp. 45-86, Birkh\"{a}user, Boston, Mass (1999).

\smallskip\noindent [18]  T.L. Gill, W.W. Zachary, and J. Lindesay,  {\it Foundations of Phys.} {\bf{31}} (2001) 1299.

\smallskip\noindent [19] M. Dresden, in {\it Renormalization: From Lorentz to Landau (and Beyond)}, L. M. Brown (ed.), Springer-Verlag, New York,  (1993).

\smallskip\noindent [20] R. P. Feynman, R. B. Leighton, and M. Sands,  {\it The Feynman Lectures on Physics}, Vol II.  Addison-Wesley, New York (1974).

\smallskip\noindent [21] M.H.L. Pryce, Proc. Roy. Soc. London  A {\bf 195}, 400 (1948).

\smallskip\noindent [22] D. G. Currie, T. F. Jordan, and E. C. G. Sudarshan,
Rev. Mod. Phys. {\bf 35}, 350  (1963).

\smallskip\noindent [23]  T.L. Gill, W.W. Zachary, and J. Lindesay,   {\it Int. J. Theor. Phys.} {\bf{37}} (1998) 2573.

\smallskip\noindent [24] L. P. Horwitz and C. Piron, Helv. Phys. Acta {\bf 46}, 316 (1981).

\smallskip\noindent [25] V. L. Arnol'd,   {\it Mathematical Methods in Classical Mechanics}, Springer-Verlag, New York  (1993).

\smallskip\noindent [26] A. A. Penzias and R. W. Wilson, {\it Ap. J.} {\bf{142}} (1965) 419.

\smallskip\noindent [27]  P.J.E. Peebles, {\it Principles of Physical Cosmology},  Princeton University Press, London (1993).

\smallskip\noindent [28] A. G. Cohen and S. L. Glashow (arXiv:hep-ph/0605036 v1).

\smallskip\noindent [29] R. M. Santilli,  {\it Elements of Hadronic Mechanics, I: Mathematical Foundations}, Ukraine Academy of Sciences, Kiev (1993).

\smallskip\noindent[30]  R.P. Feynman, {\it Phys. Rev.} {\bf{74}} (1948) 939.

\smallskip\noindent[31]  E.C.G. Stueckelberg, {\it Helv. Phys. Acta} {\bf{15}} (1942) 23.

\smallskip\noindent [32] H.  Paul,  {\it Photonen: Experimente und ihre Deutung}, Vieweg, Braunschweig (1993).

\smallskip\noindent [33] R. J. Buenker, J. Chem. Phys. (Russia) {\bf 22}, (2003) 124.

\smallskip\noindent [34] V. Bargmann and E. P. Wigner, Proc. Nat. Acad. Sci. {\bf 34}, (1948) 211.

\smallskip\noindent [35] E. Schr\"odinger and L. Bass, Proc. R. Soc. London {\bf A232},  (1938) 1.

\smallskip\noindent [36] A. Goldhaber and M. Nieto,  Rev. Mod. Phys. {\bf 43}, (1971) 277.

\smallskip\noindent [37] R. Jackiw, Comments Mod. Phys. {\bf 1A}, (1999) 1.

\smallskip\noindent [38] R. P. Feynman, {\it Quantum Electrodynamics}, W. A. Benjamin, New York  (1964)

\smallskip\noindent [39] A. I. Akhiezer and V. D. Berestetskii, {\it Quantum Electrodynamics}, Wiley-Interscience, New York (1965) 

\smallskip\noindent [40] R. V. Pound and J. L. Snider, Phys. Rev.  {\bf 140}, B788 (1965).

\smallskip\noindent [41] H. D. Curtis, {\it Publ. Lick Obs.} {\bf{13}}(1918).

\smallskip\noindent [42] T. J. Pearson and J. A. Zensus, in {\it Superluminal Radio Sources}, (eds Zensus \& Pearson) Cambridge University Press, London (1987). 

\smallskip\noindent [43] J. A. Zensus, {\it Annu. Rev. Astron. Astrophys.} {\bf{35}} (1997) 607.

\smallskip\noindent [44] I. F. Mirabel and L. F. Rodr\'{i}guez, {\it Ann. Rev. Astron. Astrophys.} {\bf{37}} (1999) 409.

\smallskip\noindent [45] M. J. Rees, {\it Nature} {\bf{211}} (1966) 468.  

\smallskip\noindent [46] A. De R\'{u}jula, (arXiv:hep-ph/0412094 v1).

\smallskip\noindent [47] K. Greisen, {\it Phys. Rev. Lett.} {\bf{16}} (1966) 748.

\smallskip\noindent [48] G. T. Zatsepin and V. A. Kuzmin, {\it Sov. Phys. JETP. Lett.} {(\it{Engl. Transl.})} {\bf{4}} (1966) 78.

\smallskip\noindent [49] M. Takeda et al, {\it Phys. Rev. Lett.} {\bf{81}} (1998) 1163; {\it Astrophys. J.} {\bf{522}} (1999) 225.

\smallskip\noindent [50] C. H. Jui (HiRes collaboration), {\it Proc. 27th International Cosmic Ray Conference 2001, Hamburg}.

\smallskip\noindent [51] A. Dar, arXiv: 096.0973v1 (2009).

\smallskip\noindent [52] C. Payne-Gaposchkin and  K. Haramundanis, {\it Introduction to Astronomy},  Prentice-Hall Inc., Englewood Cliffs, NJ, second Edition, (1970). 

\smallskip\noindent [53] T. L. Gill and W. W. Zachary, {\it J. Math. Phys. } {\bf{43}} (2002) 69-93.

\smallskip\noindent [54] R. P. Feynman, {\it Phys. Rev. } {\bf{81}} (1951) 108-128.

\smallskip\noindent [55] F. J. Dyson, {\it Phys. Rev. } {\bf{75}} (1949) 1736-1755.

\smallskip\noindent [56] T. L. Gill and W. W. Zachary, {\it J. Phys. A: Math.. and Gen.} {\bf{38}} (2005) 2479-2496; Corrigendum: Ibid., {\bf{39}} (2006), 1537-1538.

\smallskip\noindent [57] T. L. Gill, W. W. Zachary and M. Alfred, {\it J. Phys. A: Math.. and Gen.} {\bf{38}} (2005) 6955-6976.

\smallskip\noindent [58] S. S. Afshar, E. Flores, K. F. McDonald and E. Knoesel, {\it Foundations of Physics} {\bf{37}} (2007) 295-305.

\smallskip\noindent [59] I. V. Drozdov and A. A. Stahlhofen, arXiv: 0803.2596v1 (2008).

\end{document}